\documentclass[twocolumn,10pt]{article}
\usepackage[letterpaper,textwidth=6.75in,textheight=9.25in,columnsep=0.25in,centering]{geometry}
\usepackage[T1]{fontenc}
\usepackage[utf8]{inputenc}
\usepackage{amsmath}
\usepackage{amssymb}
\usepackage{amsthm}
\usepackage[nopatch]{microtype}
\usepackage{booktabs}
\usepackage{graphicx}
\usepackage{xcolor}
\usepackage{xurl}
\usepackage[round,authoryear]{natbib}
\usepackage[hidelinks,pdftitle={Cluster-Robust Prediction-Powered Inference},pdfauthor={David Broska, Michael Howes}]{hyperref}
\newcommand{\PPI}{\mathrm{PP}}
\newcommand{\Expect}{\mathbb{E}}
\newcommand{\prob}{\mathbb{P}}
\newcommand{\reals}{\mathbb{R}}
\newcommand{\Tr}{\mathop{\bf Tr}}
\newcommand{\argmin}{\mathop{\rm argmin}}

\DeclareMathOperator*{\var}{Var}
\DeclareMathOperator*{\cov}{Cov}

\theoremstyle{definition}
\newtheorem{theorem}{Theorem}
\newtheorem{proposition}{Proposition}
\newtheorem{lemma}{Lemma}

\theoremstyle{definition}

\newtheorem{example}{Example}
\newtheorem{assumption}{Assumption}
\theoremstyle{remark}
\newtheorem{remark}{Remark}

\begin{document}

\twocolumn[{
\begin{center}
  {\LARGE\bfseries Cluster-Robust Prediction-Powered Inference\par}
  \vspace{1.2em}
  \begin{tabular}{c@{\hspace{4em}}c}
    \large David Broska & \large Michael Howes \\[3pt]
    Department of Sociology & Department of Statistics \\
    Stanford University & Stanford University \\
    \texttt{dbroska@stanford.edu} & \texttt{mhowes@stanford.edu}
  \end{tabular}\par
  \vspace{1.2em}
  \begin{minipage}{0.92\textwidth}
    \small\textbf{Abstract.} Data collection is often costly or logistically demanding, limiting both the questions researchers can pursue and how precisely they can answer them. Prediction-powered inference (PPI) can reduce the amount of data needed for precise parameter estimation by combining labeled data with machine learning predictions. However, ignoring dependence within clusters can produce confidence intervals that cover the true parameter less often than their nominal rate. We introduce \textit{Cluster-Robust PPI++}, which provides standard errors in closed form and asymptotically valid confidence intervals under arbitrary dependence within independent clusters, requiring no bootstrap or resampling. Our central contribution is to accommodate \textit{partially labeled clusters}, a common empirical setting in which clusters contain both labeled and unlabeled units. As units are dependent within clusters, partially labeled clusters violate the independence assumption of PPI++. We also show how precision increases depend on the labeling design, and derive a \textit{cluster-aware power tuning} rule that minimizes asymptotic variance. In an application to television news, standard PPI++ confidence intervals have coverage below 60\%, whereas Cluster-Robust PPI++ can achieve nominal 95\% coverage.\par
    \vspace{0.5em}
    \textbf{Keywords:} prediction-powered inference, cluster-robust standard errors, partially labeled clusters, power tuning, machine learning, computational social science
  \end{minipage}
\end{center}
\vspace{1.5em}
}]

\section{Introduction}

Collecting high-quality data often requires substantial resources, whether to recruit study participants, obtain expert annotations, or coordinate data collection with specialized instruments across settings and over time. As a result, some research questions are too costly to investigate at scale, while studies that are feasible may rely on too little data to estimate quantities of interest precisely \citep{marek_reproducible_2022, arel-bundock_quantitative_2026}. A growing literature examines how predictions from artificial intelligence models can help researchers make more efficient use of limited data \citep{egami_using_2023, hullman_this_2026, ludwig_large_2026, van_loon_using_2026}.

Prediction-powered inference (PPI) is a statistical method that reduces the data needed for precise parameter estimation by combining labeled data with machine learning predictions \citep{angelopoulos_prediction-powered_2023-1}. This method yields valid confidence intervals even when the predictions are arbitrarily biased. PPI++ extends this approach through power tuning, which adjusts the weight assigned to predictions to minimize asymptotic variance \citep{angelopoulos_ppi_2024}.

PPI++ assumes independence across sampled units, but many applications involve data that are correlated within clusters. Examples include images from the same event \citep{boussalis_gender_2021, rister_portinari_maranca_correcting_2025}, ratings from the same person \citep{broska_mixed_2025, mehrotra_multi-perspective_2026}, and satellite measurements from the same region \citep{kluger_prediction-powered_2025, salerno_spatially_2026}. Moreover, PPI++ assumes that the labeled and unlabeled datasets are independent, which may not hold when both datasets contain units from the same cluster.

Just as ignoring prediction error invalidates classical confidence intervals~-- the very problem PPI++ is designed to solve~-- ignoring within-cluster dependence undermines the coverage guarantees PPI++ is meant to provide. Applied to clustered data, standard PPI++ standard errors are invalid in that its confidence intervals cover the truth far less often than their nominal rate.

We introduce Cluster-Robust PPI++, which extends PPI++ to data with arbitrary within-cluster dependence. We derive closed-form, cluster-robust sandwich variance estimators that yield valid standard errors and confidence intervals with no bootstrap or resampling. We also extend PPI++ power tuning to the clustered setting, so the tuning parameter adapts to both the accuracy of the predictions and the degree of within-cluster dependence. The tuned estimator has asymptotic variance no greater than that of the corresponding estimator based solely on labeled data.

Our central contribution is inference with \emph{partially labeled clusters}, in which labeled and unlabeled units coexist within the same cluster. This gives researchers the flexibility to sample labels at the unit level rather than sampling the labels for entire clusters. Existing cluster-aware approaches to PPI \citep{kluger_prediction-powered_2025} require each cluster to be either fully labeled or fully unlabeled and rely on the bootstrap; our closed-form estimator removes both restrictions. Figure~\ref{fig:cluster-data} illustrates this data structure, showing labeled and unlabeled units within clusters of different sizes.

\begin{figure*}[t]
    \centering
    \includegraphics[width=\textwidth]{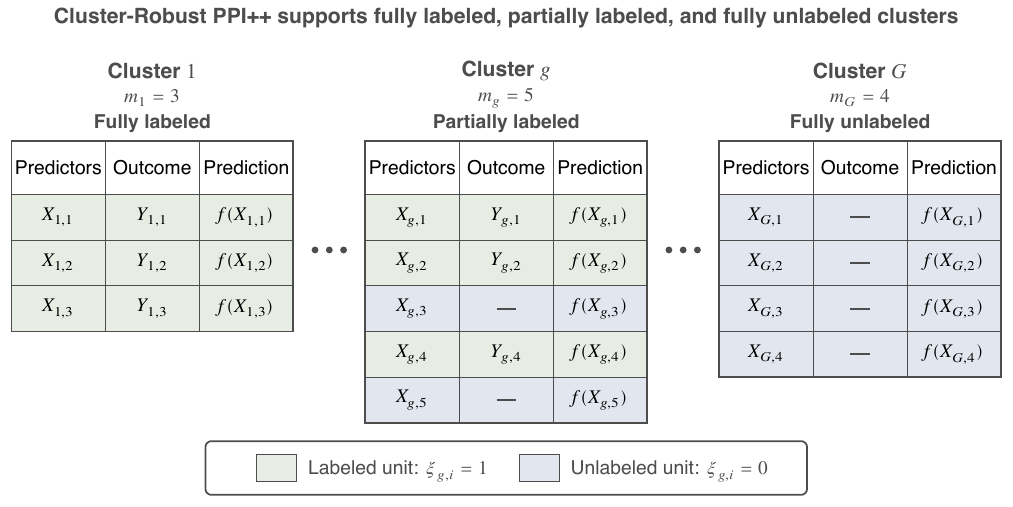}
    \caption{Clustered data for prediction-powered inference. Each row represents
    unit \(i\) in cluster \(g\). Predictors \(X_{g,i}\) and predictions
    \(f(X_{g,i})\) are available for every unit. Outcomes \(Y_{g,i}\) are
    observed for all, some, or none of the units, making the cluster fully
    labeled, partially labeled, or fully unlabeled, respectively.}
    \label{fig:cluster-data}
\end{figure*}

Moreover, the clustered setting introduces design choices that are absent from standard PPI++. With i.i.d.\ data, researchers mainly choose how many units to label relative to the number of unlabeled, predicted units. With clustered data, precision also depends on how labels are allocated across the cluster structure: many labels in a few clusters versus a few labels in many clusters, and the extent of within-cluster dependence in outcomes and predictions. We derive formulas that show how these choices affect cluster-robust standard errors. These formulas can help researchers compare labeling designs and identify the most precise option.

We evaluate Cluster-Robust PPI++ in an application to television news, in which a face recognition system identifies politicians on screen, using both fully labeled and partially labeled designs. In this application, the method substantially improves confidence-interval coverage relative to standard PPI++ while retaining gains in precision over estimates based solely on labeled data.

\section{Methods}\label{sec:methods}

We present cluster-robust standard errors for the PPI++ estimator of the mean of \(Y\), since it is the most illustrative case. Appendix~\ref{app:methods} derives cluster-robust standard errors for general PPI++ M-estimators. The PPI++ estimator \citep{angelopoulos_ppi_2024} is defined in Section~\ref{sec:ppi++}. Section~\ref{sec:DGP} describes the assumptions underlying the derivations of cluster-robust standard errors. The key assumptions are that clusters are drawn i.i.d. and that within each cluster, every label has the same probability of being sampled. The cluster-robust standard errors are presented in Section~\ref{sec:cluster-se}. Section~\ref{sec:power-tuning} and \ref{sec:gain} contains a discussion on power tuning and the PPI correlation for clustered data \citep{broska_mixed_2025}. 
The design decisions for Cluster-Robust PPI++ are discussed in Section~\ref{sec:design}. Finally, Section~\ref{sec:general} briefly discusses how the results for mean estimation generalize to other estimators.

\subsection{The PPI++ estimator}\label{sec:ppi++}

The PPI++ estimator was introduced in \citet{angelopoulos_ppi_2024}. For estimating the mean, the estimator PPI++ is equal to the generalized regression (GREG) estimator used in survey sampling \citep{sarndal2003model, mozer2026ppi}. The estimator is equal to
\begin{equation}\label{eq:PPI++mean}
    \hat\theta^{\PPI}_\lambda = \frac{\lambda}{N}\sum_{i=1}^N f(\widetilde{X}_i) + \frac{1}{n}\sum_{i=1}^n \left(Y_i-\lambda f(X_i)\right),
\end{equation}
where \((X_i,Y_i)_{i=1}^n\) is a labeled dataset and \((\widetilde{X}_i)_{i=1}^N\) is an unlabeled dataset. Both datasets are assumed to be independent samples with the same predictor distribution: \(X_i \stackrel{d}{=} \widetilde{X}_i\). Within each dataset, units are sampled i.i.d. The function \(f\) is an AI model that predicts \(Y\) from \(X\) and is fixed independently of both the labeled and unlabeled datasets. Finally, \(\lambda\) is a power tuning parameter which is chosen to minimize the asymptotic variance of \(\hat\theta_\lambda^{\PPI}\). Note that when \(\lambda=0\), the PPI++ estimator is the sample mean of \((Y_i)_{i=1}^n\).

Under the above assumptions, the estimator \(\hat \theta^{\PPI}_\lambda\) is unbiased for \(\theta^\star = \Expect[Y]\). Furthermore, if \(X_i\) and \(Y_i\) both have finite variance, then the central limit theorem implies if \(n,N \to \infty\) and \(n/N \to r\), then
\begin{equation}\label{eq:iid-CLT}
    \sqrt{n}(\hat\theta^\PPI_\lambda - \theta^\star) \stackrel{d}{\rightarrow} \mathcal{N}(0,\sigma^2_\lambda),
\end{equation}
where \(\mathcal{N}(\mu,\sigma^2_\lambda)\) denotes the univariate normal distribution and \(\stackrel{d}{\to}\) denotes convergence in distribution. \citet{angelopoulos_ppi_2024} derive a consistent estimator for \(\sigma^2_\lambda\) as well as the value \(\lambda^\star\) that minimizes \(\sigma^2_\lambda\). They provide a plug-in estimator \(\hat\lambda\) for \(\lambda^\star\) and show that ``power-tuned estimator'' \(\hat\theta^\PPI_{\hat\lambda}\) achieves the same asymptotic variance as \(\hat\theta^\PPI_{\lambda^\star}\). The power-tuned estimator is no longer unbiased but has bias  \(O(n^{-1})\). Since the bias is lower order, the following central limit theorem still holds:
\begin{equation*}
    \sqrt{n}(\hat\theta^\PPI_{\hat\lambda} - \theta^\star) \stackrel{d}{\rightarrow} \mathcal{N}(0,\sigma^2_{\lambda^\star}).
\end{equation*}

This central limit theorem can be used to construct asymptotically valid confidence intervals for \(\theta^\star\). Furthermore, since \(\sigma^2_{\lambda^\star} \le \sigma^2_0\), the power tuned PPI++ estimator \(\hat\theta^\PPI_{\hat\lambda}\) is always at least as asymptotically efficient as the labels-only estimator \(\hat\theta^\PPI_0\).

The derivations in \citet{angelopoulos_ppi_2024} heavily rely on the assumptions that \(((X_i,Y_i))_{i=1}^n\) and \((\widetilde{X}_i)_{i=1}^N\) are i.i.d. and that these two datasets are independent. Neither of these assumptions holds in the setting of clustered data which we will describe next.

\subsection{Data generating process for clustered data}\label{sec:DGP}

This section introduces the notation for Cluster-Robust PPI++ and describes its two key assumptions about the cluster structure. 
The \(g\)th cluster will be denoted by \((X_{g,i},Y_{g,i})_{i=1}^{m_g}\). The first assumption is that each cluster is an i.i.d. draw from a cluster distribution. This means that units \((X_{g,i},Y_{g,i})\) are independent across clusters but may be arbitrarily dependent within each cluster. The size of the cluster, \(m_g\), is a random variable and is allowed to be correlated with \((X_{g,i},Y_{g,i})\). The  number of clusters is denoted by \(G\).

The second assumption is that there exists a cluster level probability \(p_g\) such that each outcome \(Y_{g,i}\) is observed with probability \(p_g\). More precisely, let \(\xi_{g,i}\) indicate whether the outcome \(Y_{g,i}\) is labeled. That is, if \(\xi_{g,i}=1\), then both \(X_{g,i}\) and \(Y_{g,i}\) are observed. Otherwise, only \(X_{g,i}\) is observed. It is assumed that conditional on \(p_g\), the label indicators \((\xi_{g,i})_{i=1}^{m_g}\) are i.i.d. Bernoulli random variables with \(\prob(\xi_{g,i}=1|p_g)=p_g\). The probabilities \((p_g)_{g=1}^G\) are assumed to be i.i.d. and it is assumed that \(p_g\) and \((\xi_{g,i})_{i=1}^{m_g}\) are independent of the cluster data \((X_{g,i},Y_{g,i})_{i=1}^{m_g}\). See Figure~\ref{fig:cluster-data} for an illustration of the labeling process in the clustered setting.

To see the connection with standard PPI++, the labeled dataset \(((X_{i},Y_i))_{i=1}^n\) is replaced with \((X_{g,i},Y_{g,i})_{\xi_{g,i}=1}\) and the unlabeled dataset \((\widetilde{X})_{i=1}^N\) is replaced with \((X_{g,i})_{\xi_{g,i}=0}\). The number of labeled units is \(n=\sum_{g=1}^G \sum_{i=1}^{m_g}\xi_{g,i}\) and the number of unlabeled units is \(N=\sum_{g=1}^G\sum_{i=1}^{m_g}(1-\xi_{g,i})\). Note that there may be units from the same cluster that are in both the labeled and unlabeled datasets. Thus, the labeled and unlabeled datasets are not independent.

Some further assumptions are placed on the cluster distribution. Specifically, it is assumed that the absolute sums \(\sum_{i=1}^{m_g}|Y_{g,i}|\) and \(\sum_{i=1}^{m_g}|f(X_{g,i})|\) have finite variance. It also assumed that the cluster size \(m_g\) has finite expectation and variance. Finally, it is required that the expected value of \(p_g\) is bounded away from \(0\) and \(1\).

\subsection{Cluster robust standard errors}\label{sec:cluster-se}

The main result in this section is Theorem~\ref{thrm:cluster-clt} which gives a central limit theorem for the PPI++ estimator under the data generating process described in the previous section. Theorem~\ref{thrm:cluster-clt} also includes a plug in estimator for the asymptotic variance of the PPI++ estimator which can be used to form asymptotically valid confidence intervals (Remark~\ref{rem:CI}).

We will assume that the estimand is the unit-level mean of \(Y\). That is
\begin{equation}
    \theta^\star = \frac{\Expect\left[\sum_{i=1}^{m_g}Y_{g,i}\right]}{\Expect[m_g]}.
\end{equation}
In general, this is a different estimand than cluster-level mean \(\Expect\left[\frac{1}{m_g}\sum_{i=1}^{m_g}Y_{g,i}\right]\) which is covered by the more general framework in Appendix~\ref{app:methods}.

Let \(\theta^f = \Expect[\sum_{i=1}^{m_g}f(X_{g,i})]/\Expect[m_g]\) be the unit-level mean of the predictions \(f(X)\). Also, let \(\kappa = \Expect[m_g]\) be the expected cluster size and \(r = \frac{\Expect[p_g]}{1-\Expect[p_g]}\) be the expected ratio of labeled to unlabeled data.

The following cluster-level random variables will be used to define the asymptotic variance of \(\theta^{\PPI}_\lambda\). For each cluster \(g=1,\ldots,G\), define
\begin{align*}
    \varepsilon_{g} & =\sum_{i=1}^{m_g}\xi_{g,i}(Y_{g,i}-\theta^\star), \\
    Z_{g}           & =\sum_{i=1}^{m_g}(\xi_{g,i}-r(1-\xi_{g,i}))(f(X_{g,i})-\theta^f).
\end{align*}
The first variable \(\varepsilon_g\) is a sum of   residuals and measures the cluster-level variability in the labels \(Y_{g,i}\). The second variable \(Z_g\) incorporates the predictions \(f(X_{g,i})\) without adding bias, and can be thought of as a cluster-level control variate.
The variables \((\varepsilon_g,Z_g)_{g=1}^G\) are both mean zero and
\begin{equation}\label{eq:ppi-approx}
    \hat{\theta}^{\PPI}_\lambda - \theta^\star \approx \frac{1}{n}\left(\sum_{g=1}^G\varepsilon_g - \lambda\sum_{g=1}^G Z_g\right)
\end{equation}
where the approximation comes from replacing \(\frac{n}{N}\) with the large sample limit \(r\). Equation \eqref{eq:ppi-approx} implies that the limiting distribution of \(\hat{\theta}^\PPI_\lambda-\theta^\star\) is determined by the limiting distribution of \((\varepsilon_g,Z_g)_{g=1}^G\).

Let \(\Sigma \in \reals^{2 \times 2}\) be the population covariance matrix for the pair \(\varepsilon_g,Z_g\), so that
\begin{equation}\label{eq:sigma-hat}
    \Sigma = \begin{bmatrix}
        \var(\varepsilon_g)      & \cov(\varepsilon_g, Z_g) \\
        \cov(\varepsilon_g, Z_g) & \var(Z_g)
    \end{bmatrix}
\end{equation}
To estimate \(\Sigma\), define the following in-sample version of \(\varepsilon_g\) and \(Z_g\)
\begin{align}
    \hat{\varepsilon}_g & =\sum_{i=1}^{m_g}\xi_{g,i}(Y_{g,i}-\hat\theta^\PPI_{\lambda}), \label{eq:hat-epsilon}\\
    \hat{Z}_g           & = \sum_{i=1}^{m_g}\left(\xi_{g,i}-\tfrac{n}{N}(1-\xi_{g,i})\right)(f(X_{g,i})-\hat\theta^f),\label{eq:hat-Z}
\end{align}
where \(\hat\theta^f = \frac{1}{n+N}\sum_{g=1}^G\sum_{i=1}^{m_g}f(X_{g,i})\). Let \(\widehat{\Sigma}\) be the following empirical estimate of \(\Sigma\),
\[
    \widehat{\Sigma} = \frac{1}{G}\sum_{g=1}^G \begin{bmatrix}
        \hat{\varepsilon}_g^2 & \hat{\varepsilon}_g\hat{Z}_g\\
        \hat{Z}_g\hat{\varepsilon}_g & \hat{Z}_g^2
    \end{bmatrix}
\]

Theorem~\ref{thrm:cluster-clt} gives the asymptotic distribution of \(\hat{\theta}^\PPI_\lambda\) as the number of clusters \(G\) goes to infinity.

\begin{theorem}[Central limit theorem for the PPI++ estimator]\label{thrm:cluster-clt}
    Suppose that the assumptions in Section~\ref{sec:DGP} hold. Then, as \(G \to \infty\),
    \[
        \sqrt{n}\left(\hat\theta^\PPI_\lambda - \theta^\star\right) \stackrel{d}{\rightarrow} \mathcal{N}(0,\sigma^2_\lambda),
    \]
    where
    \begin{align}\label{eq:asymp-variance}
        \sigma^2_\lambda & =\frac{1+r}{r\kappa}\left(\var(\varepsilon_g) - 2\lambda \cov(Z_g, \varepsilon_g) + \lambda^2 \var(Z_g)\right).
    \end{align}
    Furthermore, if
    \begin{equation}\label{eq:variance-est}
        \hat\sigma^2_\lambda = \frac{G}{n}\left(\widehat{\Sigma}_{1,1}-2\lambda\widehat{\Sigma}_{1,2}+\lambda^2\widehat{\Sigma}_{2,2}\right),
    \end{equation}
    then \(\hat\sigma^2_\lambda\) converges to \(\sigma^2_\lambda\) in probability and thus
    \begin{equation}\label{eq:cluster-clt}
        \frac{\sqrt{n}}{\hat\sigma_\lambda}\left(\hat\theta^\PPI_\lambda - \theta^\star\right) \stackrel{d}{\rightarrow} \mathcal{N}(0,1).
    \end{equation}
\end{theorem}
\begin{remark}\label{rem:CI}
    Let \(z_{1-\alpha/2}\) be the \(1-\alpha/2\) quantile for the standard normal distribution. Theorem~\ref{thrm:cluster-clt} implies that the interval
    \[\mathcal{C}_\alpha^\PPI = \left(\hat\theta^\PPI_\lambda \pm z_{1-\alpha/2}\hat\sigma_\lambda/\sqrt{n} \right),\]
    is an asymptotically valid confidence interval for \(\theta^\star\). Our open source implementation of Cluster-Robust PPI++ uses this confidence interval construction.
\end{remark}

Theorem~\ref{thrm:cluster-clt} is proved in the appendix. Example~\ref{ex:no-cluster} shows that Theorem~\ref{thrm:cluster-clt} can be used to recover the central limit theorem in \citet{angelopoulos_ppi_2024} for non-clustered data.

\begin{example}
    \label{ex:no-cluster}
    Consider the setting of \citet{angelopoulos_ppi_2024} where there pairs \((X_i,Y_i)\) are i.i.d. For i.i.d. data, the number of ``clusters'' \(G\) is equal to the total number of data points \(n+N\) and the label indicators \(\xi_i\) are i.i.d Bernoulli random variables with expectation \(\frac{r}{1+r}\). The expected cluster size \(\kappa\) is equal to 1.

    In the i.i.d. setting, we have
    \[
        (\varepsilon_g, Z_g) = \begin{cases}
            (Y_{g}-\theta^\star, f(X_{g})-\theta^f) & \text{if } \xi_g=1, \\
            (0, -r(f(X_{g})-\theta^f))              & \text{if } \xi_g=0.
        \end{cases}
    \]
    A direct calculation gives that
    \begin{align*}
        \Sigma = \frac{r}{1+r}\begin{bmatrix}
                                  \var(Y_g)         & \cov(Y_g, f(X_g)) \\
                                  \cov(Y_g, f(X_g)) & (1+r)\var(f(X_g))
                              \end{bmatrix}.
    \end{align*}
    And thus,
    \begin{align*}
        \sigma^2_\lambda & = \var(Y_g) - 2\lambda \cov(Y_g, f(X_g))              \\
                         & \quad{}+ \lambda^2(1+r)\var(f(X_g))                   \\
                         & =\var(Y_g-\lambda f(X_g)) + \lambda^2 r \var(f(X_g)),
    \end{align*}
    which agrees with Theorem~1 in \citet{angelopoulos_ppi_2024}.
\end{example}

\subsection{Power tuning}\label{sec:power-tuning}

The parameter \(\lambda\) can be adaptively tuned to minimize the asymptotic variance \(\sigma_\lambda^2\). \citet{angelopoulos_ppi_2024} call this procedure \emph{power tuning}. A simple calculation gives
\begin{equation}
    \lambda^\star := \argmin\{ \sigma^2_\lambda: \lambda \in \reals\} = \frac{\cov(\varepsilon_g, Z_g)}{\var(Z_g)}.\label{eq:lambda-star}
\end{equation}
The optimal \(\lambda^\star\) can be estimated using the empirical covariance \(\widehat{\Sigma}\) from \eqref{eq:sigma-hat}. This gives
\begin{equation}
    \hat{\lambda} := \frac{\widehat{\Sigma}_{1,2}}{\widehat{\Sigma}_{2,2}}.\label{eq:lambda-hat}
\end{equation}
 Specifically, following \cite{angelopoulos_ppi_2024}, we recommend first computing \(\hat\theta^\PPI_1\) and using this estimate to compute \(\hat{\varepsilon}_g\) and \(\hat{Z}_g\). These variables are then used to compute \(\hat\lambda\) via \eqref{eq:lambda-hat} which then gives the final estimate \(\hat \theta^\PPI_{\hat{\lambda}}\).

Under the above scheme, the plug-in estimate \(\hat{\lambda}\) converges in probability to \(\lambda^\star\). This implies that, under the assumptions of Theorem~\ref{thrm:cluster-clt},
\begin{equation}\label{eq:power-tuned-clt}
    \sqrt{n}(\hat\theta^\PPI_{\hat\lambda}-\theta^\star) \stackrel{d}{\rightarrow} \mathcal{N}(0,\sigma^2_{\lambda^\star}).
\end{equation}
Thus, the power tuned estimator \(\hat\theta^\PPI_{\hat\lambda}\) has asymptotic variance \(\sigma^2_{\lambda^\star} \le \sigma^2_0\) and so the power tuned PPI++ estimator is asymptotically at least as precise as the labels-only estimator \(\hat\theta^{\PPI}_0\). Again, the following example shows the connection between the current work and the i.i.d. setting of \citet{angelopoulos_ppi_2024}.

\begin{example}\label{ex:no-cluster-power-tune}
    Suppose that each cluster only contains a single data point as in Example~\ref{ex:no-cluster}. Then, using the formula for \(\Sigma\) given there
    \begin{align*}
        \lambda^\star & = \frac{\cov(\varepsilon_g, Z_g)}{\var(Z_g)}   = \frac{\cov(Y_g, f(X_g))}{(1+r)\var(f(X_g))},
    \end{align*}
    which agrees with Example~6.1 in \citet{angelopoulos_ppi_2024}.
\end{example}
The power-tuned estimator \(\hat\theta^\PPI_{\hat\lambda}\) is used in all the experiments in Section~\ref{sec:experiments}.

\subsection{The PPI correlation}\label{sec:gain}

Let \(\hat\theta^\PPI=\hat\theta^\PPI_{\hat\lambda}\) be the power-tuned estimator. By definition, the asymptotic variance of \(\hat\theta^\PPI\) is equal to
\begin{align*}
    \sigma^2_{\lambda^\star} & = \frac{1+r}{r\kappa}\left(\var(\varepsilon_g) - \frac{\cov(\varepsilon_g, Z_g)^2}{\var(Z_g)} \right) \\
                             & =\sigma_0^2\left(1 - \rho(\varepsilon_g,Z_g)^2\right),
\end{align*}
where \(\sigma_0^2 = \frac{1+r}{r\kappa}\var(\varepsilon_g)\) is the variance of the labels-only estimator \(\hat\theta^\PPI_0\) and \(\rho(\varepsilon_g,Z_g)\) is the correlation between \(\varepsilon_g\) and \(Z_g\) measured across clusters. Thus, the power tuned PPI estimator will have substantially lower variance than the labels-only estimator when \(\rho(\varepsilon_g,Z_g)\) is close to \(\pm1\). In the case of i.i.d. sampling, \(\rho(\varepsilon_g,Z_g)\) simplifies to \(\frac{1}{\sqrt{1+r}}\rho(Y,f(X))\) as in \citet{broska_mixed_2025}. In the clustered setting, the correlation \(\rho(\varepsilon_g,Z_g)\) is harder to interpret. The correlation \(\rho(\varepsilon_g,Z_g)\) depends on the proportion of labeled units, the proportion of partially labeled clusters, and the correlations \(\rho(Y_{g,i},f(X_{g,i}))\) and \(\rho(Y_{g,i}, f(X_{g,j}))\) where \(i\) and \(j\) denote two distinct units in the same cluster.

Further formulas for the PPI variance and the correlation \(\rho(\varepsilon_g,Z_g)\) are given in Section~\ref{sec:appn-var} of the appendix. These formulas can be used to compare different labeling strategies as described in the next section.

\subsection{Design choices}\label{sec:design}

With clustered data, researchers must decide both how many units to label and how to distribute those labels across clusters. Obtaining additional labels within a sampled cluster may cost less than obtaining labels from a new cluster. Yet the same number of labels can yield different precision depending on their allocation, so researchers must consider both the cost and the statistical accuracy of their labeling choices.

Appendix~\ref{sec:appn-var} provides formulas for comparing the asymptotic variance of alternative labeling designs. These formulas show that the asymptotic variance depends on three \emph{design parameters}: the number of clusters \(G\), the expected proportion of labeled units \(\mu_p\) (or mean labeling probability for short), and the intraclass correlation between label indicators in the same cluster
\(\eta_p\). In addition to these design parameters, the variance of the PPI++ estimator also depends on certain non-design moments that capture the variability and correlation of the labels and predictions. The non-design moments are defined in the Appendix~\ref{sec:appn-var}.

We suggest that researchers could estimate the non-design moments from pilot or suitable existing data, then assess the precision of candidate designs alongside their costs. Section~\ref{sec:experiments} illustrates how the design parameters \(\mu_p\) and \(\eta_p\) affect the RMSE of the PPI++ estimator in an empirical example. These simulations also verify the formulas in Appendix~\ref{sec:appn-var}.

\subsection{General estimators}\label{sec:general}

The ideas in the previous section extend naturally to general M-estimators. Let \(\ell_\theta(X,Y)\) be a loss function and define
\begin{equation}
    \theta^\star = \argmin_{\theta \in \reals^d}\frac{1}{\kappa} \Expect\left[\sum_{i=1}^{m_g}\ell_\theta(X_{g,i},Y_{g,i})\right].
\end{equation}
The PPI++ estimator for \(\theta^\star\) is
\begin{align*}
    \hat\theta^\PPI_\lambda & =\argmin_{\theta \in \reals^d}\left(L_n(\theta) - \lambda(L_n^f(\theta)-\widetilde{L}_N^f(\theta))\right),
\end{align*}
where
\begin{align*}
    L_n(\theta)               & =\frac{1}{n}\sum_{g=1}^G\sum_{i=1}^{m_g}\xi_{g,i}\ell_\theta(X_{g,i},Y_{g,i}),        \\
    L_n^f(\theta)             & =\frac{1}{n}\sum_{g=1}^G\sum_{i=1}^{m_g}\xi_{g,i}\ell_\theta(X_{g,i},f(X_{g,i})),     \\
    \widetilde{L}_N^f(\theta) & =\frac{1}{N}\sum_{g=1}^G\sum_{i=1}^{m_g}(1-\xi_{g,i})\ell_\theta(X_{g,i},f(X_{g,i})).
\end{align*}
Under appropriate regularity conditions and the assumptions of Section~\ref{sec:DGP}, the estimator \(\hat\theta^\PPI_\lambda\) is asymptotically normal as \(G \to \infty\). The asymptotic variance of \(\hat\theta^\PPI_\lambda\) is given by a ``sandwich'' form where the ``meat'' is of a similar form as \(\sigma^2_\lambda\) in \eqref{eq:asymp-variance}. The asymptotic variance can be estimated using plug-in estimates. The plug-in estimates can be used to power-tune \(\lambda\) and choose the variance minimizing PPI-estimator.

\section{Empirical Application}\label{sec:experiments}

Images and videos are an increasingly important source of data for studying social and political communication \citep{webb_williams_images_2020,boussalis_gender_2021,guilbeault_images_2024}.
Analyzing these data requires labels that translate visual content into variables of interest, such as whether a political figure appears in an image or video frame.
Human coding can provide high-quality labels but is costly to extend to large archives.
Computer vision algorithms offer a scalable alternative, but prediction errors can introduce bias into parameter estimates \citep{rister_portinari_maranca_correcting_2025}.
Video data pose an additional challenge because frames from the same recording can be highly correlated.
Treating these frames as independent can therefore understate uncertainty.
Using Japanese television news data from \citet{girbau_face_2024}, we illustrate how researchers can use Cluster-Robust PPI++ for precise parameter estimation while accounting for prediction errors and clustering.

\subsection{Simulation design}\label{sec:sim-design}

\begin{figure*}[t]
\centering
\includegraphics[width=\textwidth]{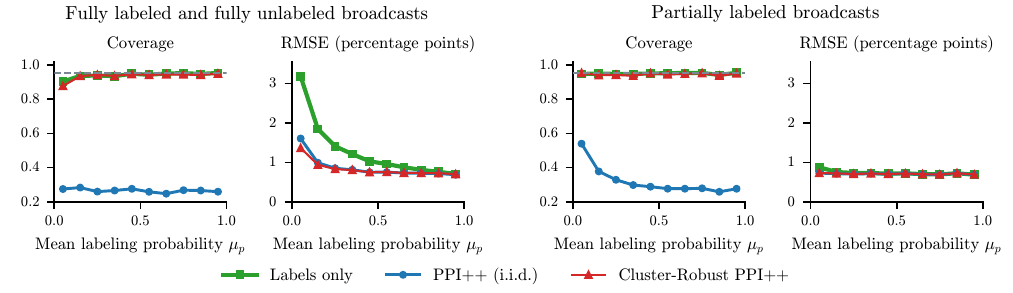}
\caption{Coverage, the share of 95\% confidence intervals that contain $\theta^\star=0.109$, and RMSE in percentage points, by the mean labeling probability $\mu_p$. In the fully labeled design (left), each broadcast is fully labeled or fully unlabeled. In the partially labeled design (right), most broadcasts are partially labeled. 2{,}500 replications per value of $\mu_p$.}
\label{fig:telestats}
\end{figure*}

The authors sampled one frame per second from excerpts of the broadcasts and manually annotated which of 41 target politicians appeared in each frame. The sampled footage comprises $M=67{,}062$ frames from $G=2{,}266$ broadcasts, where a broadcast is one airing of one program on a particular date. At least one target politician appears in $19{,}812$ of these frames; the remaining frames show none of them. We treat this collection of sampled frames and broadcasts as an artificial study population. The human annotations provide an observed outcome for every frame, alongside the corresponding automated prediction. This allows us to establish the population ground truth and then evaluate estimators that have access to only a subset of the human annotations.

Let $Y_{g,i}=1$ if the sitting prime minister (Shinz\={o} Abe, Yoshihide Suga, or Fumio Kishida, depending on the date) appears in sampled frame $i$ of broadcast $g$, and $Y_{g,i}=0$ otherwise. Treating the human annotations as correct, the target is the share of sampled seconds in which the prime minister is on screen,
\[
\theta^\star=\frac{1}{M}\sum_{g=1}^{G}\sum_{i=1}^{m_g}Y_{g,i}=0.109,
\]
where $m_g$ is the number of sampled frames in broadcast $g$ and $M=\sum_{g=1}^{G}m_g$. When broadcasts are drawn at random from the study population, $\theta^\star$ is the unit-level mean of Section~\ref{sec:cluster-se}. Each sampled second receives equal weight, so the target measures screen time in the sampled footage rather than in complete broadcasts.

The predictions come from the authors' face recognition system, whose output for the full archive they released with their replication data. The system detects faces with pretrained models, links detections of the same face across frames, and assigns an identity by comparing the faces with a few template images of each politician collected from the internet. We set $f(X_{g,i})=1$ if the released output identifies the sitting prime minister in frame $i$ of broadcast $g$, and $f(X_{g,i})=0$ otherwise.

Frames are nested in broadcasts, which contribute between 3 and 438 sampled frames (median 18). We treat each broadcast as a cluster, which allows both outcomes and prediction errors to be correlated within broadcasts. The intraclass correlation within broadcasts is $0.71$ for the outcome and $0.50$ for the prediction error $Y_{g,i}-f(X_{g,i})$.

We evaluate the estimators by repeatedly sampling from the study population. In each of 2{,}500 replications, we draw $G=2{,}266$ broadcasts with replacement and treat each draw as a separate cluster, so that clusters are i.i.d.\ as in Section~\ref{sec:DGP}. We then reveal the human labels of a subset of frames under two designs, each with mean labeling probability $\mu_p$. In the \emph{fully labeled} design, each broadcast is labeled with probability $\mu_p$, including all of its frames, so no broadcast is partially labeled. In the \emph{partially labeled} design, each frame is labeled with probability $\mu_p$, so most broadcasts are partially labeled. We vary $\mu_p$ from $0.05$ to $0.95$ in steps of $0.10$. In Figure~\ref{fig:telestats}, we compare three estimators. The first is the labels-only estimator, the mean of the labeled outcomes with a cluster-robust standard error. The second is PPI++ with the i.i.d.\ standard error of \citet{angelopoulos_ppi_2024}, and the third is Cluster-Robust PPI++.

\subsection{Results on coverage and RMSE}\label{sec:results}

Figure~\ref{fig:telestats} shows coverage and RMSE for both the fully labeled and partially labeled designs, and Appendix~\ref{app:telestats} reports all values. PPI++ with i.i.d.\ standard errors undercovers in both designs because frames are treated as independent despite their strong correlation. The confidence intervals are roughly one-third to one-sixth as wide as the cluster-robust intervals and thus contain $\theta^\star$ in only 25\% to 54\% of replications. Cluster-Robust PPI++ achieves coverage between 93\% and 95\% in every setting but one. In the fully labeled design with $\mu_p=0.05$, the number of labeled clusters is small, and both Cluster-Robust PPI++ (88\%) and the labels-only estimator (90\%) undercover.

The precision gains depend on the labeling design. In the fully labeled design, Cluster-Robust PPI++ reduces the RMSE of the labels-only estimator most when labels are scarce, by about half for $\mu_p$ up to $0.15$, and the reduction shrinks steadily to 10\% at $\mu_p=0.75$. In the partially labeled design, the gains are small (17\% at $\mu_p=0.05$, at most 5\% for larger $\mu_p$), because most of the remaining error comes from which broadcasts are sampled. Appendix~\ref{app:regression} shows that the same pattern holds for the coefficients of a multiple regression.

We also note that relying on the predictions alone is the least defensible option for our application. When we average the predictions of all sampled frames and ignore the labels, the resulting confidence intervals cover the true parameter in only about 20\% of replications, even after accounting for dependence within broadcasts (Appendix~\ref{app:telestats}). This is the lowest coverage of any estimator. The reason is that the face recognition system, although its recall is high at 82\%, misses 18\% of the prime minister's appearances. As a result, the estimate is biased downward by about 1.8 percentage points, about a sixth of $\theta^\star$.

\subsection{Choosing a labeling design}\label{sec:labeling-design}

\begin{figure}[t]
\centering
\includegraphics[width=\columnwidth]{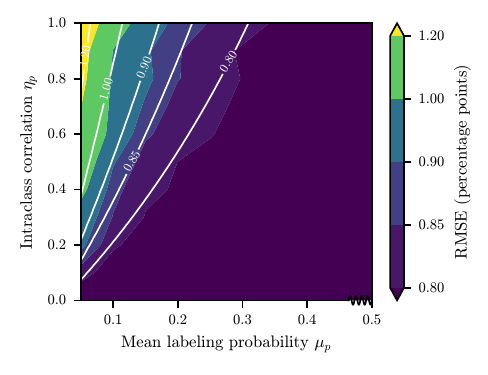}
\caption{RMSE of Cluster-Robust PPI++ in percentage points, by the mean labeling probability $\mu_p$ and the correlation $\eta_p$ between labeling indicators within a broadcast. Colors show the simulation with 2{,}500 replications per design, and white lines show the RMSE implied by the variance formula in Appendix~\ref{sec:appn-var}, at the same levels. $\eta_p=0$ is the partially labeled design and $\eta_p=1$ the fully labeled design of Figure~\ref{fig:telestats}. The horizontal axis stops at $\mu_p=0.5$ because the RMSE varies little for larger $\mu_p$.}
\label{fig:designs}
\end{figure}

Figure~\ref{fig:telestats} compares two extreme designs. Between them are designs in which broadcasts differ in how many of their frames are labeled. As in Section~\ref{sec:DGP}, each broadcast $g$ has a labeling probability $p_g$, and each of its frames is labeled with probability $p_g$. We describe each design by its mean labeling probability $\mu_p$ and by $\eta_p$, the correlation between labeling indicators within a broadcast (Appendix~\ref{sec:appn-var}). The two endpoints correspond to independent frame labeling ($\eta_p=0$) and whole-broadcast labeling ($\eta_p=1$). For intermediate designs, we draw each broadcast's labeling probability from a beta distribution chosen to give the desired $\mu_p$ and $\eta_p$. We simulate 209 designs, each with 2{,}500 replications that share the same draws of broadcasts, and Figure~\ref{fig:designs} shows those with $\mu_p$ up to 0.5. For a given share of labeled frames, the RMSE of Cluster-Robust PPI++ is smallest, up to Monte Carlo noise, when the labels are spread across all broadcasts and grows as they are concentrated in fewer broadcasts. The difference is large when few frames are labeled. At $\mu_p=0.05$, concentrating the labels in whole broadcasts nearly doubles the RMSE. When at least half of the frames are labeled, it raises the RMSE by at most 7\%.

The same designs let us check the variance formula of Appendix~\ref{sec:appn-var}, which researchers can use to plan a labeling design. To evaluate the formula, we use all human labels and predictions in the artificial study population to calculate its six moments, which summarize variation and association in outcomes and predictions. The simulated estimators still receive only the labels revealed under each design. This comparison therefore tests the variance approximation with known moments; it does not test how accurately a limited pilot estimates those moments. The simulated RMSE is within 5\% of the formula, shown as white lines in Figure~\ref{fig:designs}, in all 209 designs. The formula also reproduces the pattern above. When at least half of the frames are labeled, the formula puts the largest increase from concentrating the labels at 6\%, against 7\% in the simulation.

\section{Conclusion}

We introduced Cluster-Robust PPI++ to help researchers use predictions for precise parameter estimation while accounting for the dependence in their data and obtaining valid confidence intervals. We derived closed-form cluster-robust standard errors for PPI++ that are asymptotically valid when clusters are fully labeled, partially labeled or fully unlabeled. A cluster-aware power tuning rule weights the predictions to minimize the variance. In our application to television news, PPI++ intervals that ignore the clustering cover the true parameter far less often than their nominal rate, while Cluster-Robust PPI++ can restore nominal coverage. The precision of parameter estimates depends on how labels are allocated. When labels are scarce, spreading them across many clusters can give much more precise estimates than labeling whole clusters. Our variance formulas predict this precision closely, allowing researchers to compare labeling designs before they collect labels.

Several limitations of our method point to opportunities for future work. First, our variance estimator has no small-sample correction. Small-sample corrections and the wild cluster bootstrap are well developed for linear regression \citep{cameron_practitioners_2015, mackinnon_cluster_2023}, and adapting them to PPI++ would make the method more reliable when few clusters are labeled. Second, our theory assumes one level of clustering. Future work could extend it to multiway clustering, such as video frames grouped both by program and by date, and to panel-type estimators with cluster-specific parameters, such as fixed effects. Third, labels must be drawn with probabilities that do not depend on the data, which rules out choosing labels based on the predictions or on labels already collected. Active labeling of clusters, building on active statistical inference \citep{li_robust_2025, gligoric_can_2025}, could direct labels to the clusters where they reduce the variance most. Finally, our variance formulas describe the precision of a labeling design but not its cost. Labeling a whole cluster can be cheaper per unit than obtaining labels from different clusters, so combining the formulas with a cost model would let researchers allocate labels across and within clusters to get the most precision for a given budget.

\subsection*{Acknowledgments}
We thank Naoki Egami, Brandon Stewart, Tijana Zrnic, and Emmanuel Cand\`es for helpful comments.

\subsection*{Code and Data Availability}
Cluster-Robust PPI++ is implemented at \url{https://github.com/Michael-Howes/ppi_py}. The code and data for reproducing the analysis in this paper will be made available upon publication.

\bibliography{references}

\begin{thebibliography}{}

\bibitem[Angelopoulos et~al., 2023]{angelopoulos_prediction-powered_2023-1}
Angelopoulos, A.~N., Bates, S., Fannjiang, C., Jordan, M.~I., and Zrnic, T.
  (2023).
\newblock Prediction-powered inference.
\newblock {\em Science}, 382(6671):669--674.

\bibitem[Angelopoulos et~al., 2024]{angelopoulos_ppi_2024}
Angelopoulos, A.~N., Duchi, J.~C., and Zrnic, T. (2024).
\newblock {PPI}++: Efficient prediction-powered inference.
\newblock arXiv:2311.01453.

\bibitem[Arel-Bundock et~al., 2026]{arel-bundock_quantitative_2026}
Arel-Bundock, V., Briggs, R.~C., Doucouliagos, H., Avi{\~n}a, M.~M., and
  Stanley, T.~D. (2026).
\newblock Quantitative political science research is greatly underpowered.
\newblock {\em The Journal of Politics}, 88(1):36--46.

\bibitem[Boussalis et~al., 2021]{boussalis_gender_2021}
Boussalis, C., Coan, T.~G., Holman, M.~R., and M{\"u}ller, S. (2021).
\newblock Gender, candidate emotional expression, and voter reactions during
  televised debates.
\newblock {\em American Political Science Review}, 115(4):1242--1257.

\bibitem[Broska et~al., 2025]{broska_mixed_2025}
Broska, D., Howes, M., and van Loon, A. (2025).
\newblock The mixed subjects design: Treating large language models as
  potentially informative observations.
\newblock {\em Sociological Methods \& Research}, 54(3):1074--1109.

\bibitem[Cameron and Miller, 2015]{cameron_practitioners_2015}
Cameron, A.~C. and Miller, D.~L. (2015).
\newblock A practitioner's guide to cluster-robust inference.
\newblock {\em Journal of Human Resources}, 50(2):317--372.

\bibitem[Egami et~al., 2023]{egami_using_2023}
Egami, N., Hinck, M., Stewart, B.~M., and Wei, H. (2023).
\newblock Using imperfect surrogates for downstream inference: Design-based
  supervised learning for social science applications of large language models.
\newblock In {\em Advances in Neural Information Processing Systems},
  volume~36, pages 68589--68601.

\bibitem[Girbau et~al., 2024]{girbau_face_2024}
Girbau, A., Kobayashi, T., Renoust, B., Matsui, Y., and Satoh, S. (2024).
\newblock Face detection, tracking, and classification from large-scale news
  archives for analysis of key political figures.
\newblock {\em Political Analysis}, 32:221--239.

\bibitem[Gligori{\'c} et~al., 2025]{gligoric_can_2025}
Gligori{\'c}, K., Zrnic, T., Lee, C., Cand{\`e}s, E., and Jurafsky, D. (2025).
\newblock Can unconfident {LLM} annotations be used for confident conclusions?
\newblock In {\em Proceedings of the 2025 Conference of the Nations of the
  Americas Chapter of the Association for Computational Linguistics: Human
  Language Technologies (Volume 1: Long Papers)}, pages 3514--3533.

\bibitem[Guilbeault et~al., 2024]{guilbeault_images_2024}
Guilbeault, D., Delecourt, S., Hull, T., Desikan, B.~S., Chu, M., and Nadler,
  E. (2024).
\newblock Online images amplify gender bias.
\newblock {\em Nature}, 626:1049--1055.

\bibitem[Hullman et~al., 2026]{hullman_this_2026}
Hullman, J., Broska, D., Sun, H., and Shaw, A. (2026).
\newblock This human study did not involve human subjects: Validating {LLM}
  simulations as behavioral evidence.
\newblock arXiv:2602.15785.

\bibitem[Kluger et~al., 2025]{kluger_prediction-powered_2025}
Kluger, D.~M., Lu, K., Zrnic, T., Wang, S., and Bates, S. (2025).
\newblock Prediction-powered inference with imputed covariates and nonuniform
  sampling.
\newblock arXiv:2501.18577.

\bibitem[Li et~al., 2025]{li_robust_2025}
Li, P., Zrnic, T., and Cand{\`e}s, E. (2025).
\newblock Robust sampling for active statistical inference.
\newblock In {\em Advances in Neural Information Processing Systems},
  volume~38, pages 68727--68756.

\bibitem[Ludwig et~al., 2026]{ludwig_large_2026}
Ludwig, J., Mullainathan, S., and Rambachan, A. (2026).
\newblock Large language models: An applied econometric framework.
\newblock {\em Annual Review of Economics}, 18(1):283--316.

\bibitem[MacKinnon et~al., 2023]{mackinnon_cluster_2023}
MacKinnon, J.~G., Nielsen, M.~{\O}., and Webb, M.~D. (2023).
\newblock Cluster-robust inference: A guide to empirical practice.
\newblock {\em Journal of Econometrics}, 232(2):272--299.

\bibitem[Marek et~al., 2022]{marek_reproducible_2022}
Marek, S., Tervo-Clemmens, B., Calabro, F.~J., Montez, D.~F., Kay, B.~P.,
  Hatoum, A.~S., et~al. (2022).
\newblock Reproducible brain-wide association studies require thousands of
  individuals.
\newblock {\em Nature}, 603(7902):654--660.

\bibitem[Mehrotra et~al., 2026]{mehrotra_multi-perspective_2026}
Mehrotra, N., Visokay, A., and Gligori{\'c}, K. (2026).
\newblock Multi-perspective {LLM} annotations for valid analyses in subjective
  tasks.
\newblock arXiv:2603.21404.

\bibitem[Mozer, 2026]{mozer2026ppi}
Mozer, R. (2026).
\newblock {PPI} is the difference estimator: Recognizing the survey sampling
  roots of prediction-powered inference.
\newblock arXiv:2603.19160.

\bibitem[Newey and McFadden, 1994]{newey1994large}
Newey, W.~K. and McFadden, D. (1994).
\newblock Large sample estimation and hypothesis testing.
\newblock In Engle, R.~F. and McFadden, D., editors, {\em Handbook of
  Econometrics}, volume~4, pages 2111--2245. Elsevier.

\bibitem[Rister Portinari~Maranca et~al.,
  2025]{rister_portinari_maranca_correcting_2025}
Rister Portinari~Maranca, A., Chung, J., Hinck, M., Wolsky, A.~D., Egami, N.,
  and Stewart, B.~M. (2025).
\newblock Correcting the measurement errors of {AI}-assisted labeling in image
  analysis using design-based supervised learning.
\newblock {\em Sociological Methods \& Research}, 54(3):984--1016.

\bibitem[Salerno et~al., 2026]{salerno_spatially_2026}
Salerno, S., Wu, Z., and McCormick, T.~H. (2026).
\newblock Spatially robust inference with predicted and missing at random
  labels.
\newblock arXiv:2603.11368.

\bibitem[S{\"a}rndal et~al., 2003]{sarndal2003model}
S{\"a}rndal, C.-E., Swensson, B., and Wretman, J. (2003).
\newblock {\em Model assisted survey sampling}.
\newblock Springer.

\bibitem[van Loon and Kanopka, 2026]{van_loon_using_2026}
van Loon, A. and Kanopka, K. (2026).
\newblock Using large language models as a source of human behavioral data in
  social science experiments.
\newblock SocArXiv preprint, doi:10.31235/osf.io/y74mu\_v1.

\bibitem[Webb~Williams et~al., 2020]{webb_williams_images_2020}
Webb~Williams, N., Casas, A., and Wilkerson, J.~D. (2020).
\newblock {\em Images as Data for Social Science Research: An Introduction to
  Convolutional Neural Nets for Image Classification}.
\newblock Cambridge University Press.

\end{thebibliography}

\clearpage
\appendix
\onecolumn
\begin{center}{\LARGE\bfseries Appendix\par}\end{center}

\section{Cluster standard errors for M-estimators}\label{app:methods}

This appendix mirrors the structure of Section~\ref{sec:methods}. First, the general PPI++ estimator is discussed, followed by assumptions on the data generating process. A central limit theorem for the PPI++ estimator is then proved which is used to construct valid cluster robust standard errors. The final two sections discuss power tuning and the PPI correlation.

\subsection{The PPI++ estimator}

The PPI++ estimator is introduced in \citet{angelopoulos_ppi_2024}. Let \(\ell_\theta(x,y)\) be a loss function for the parameter \(\theta \in \reals^{d}\). This loss function defines a parameter \(\theta^\star\) via
\begin{equation*}
    \theta^\star =\argmin_\theta \Expect[\ell_\theta(X_i,Y_i)],
\end{equation*}
where \((X_i,Y_i)\) represents a single data point. Let \((X_i,Y_i)_{i=1}^n\) be a dataset of labeled data and let \((\widetilde{X}_i)_{i=1}^N\) be a dataset of unlabeled data. Define
\begin{align*}
    L_n(\theta)               & = \frac{1}{n}\sum_{i=1}^n \ell_\theta(X_i,Y_i),                            \\
    L_n^f(\theta)             & = \frac{1}{n}\sum_{i=1}^n \ell_\theta(X_i, f(X_i)),                        \\
    \widetilde{L}_N^f(\theta) & = \frac{1}{N}\sum_{i=1}^N \ell_\theta(\widetilde{X}_i, f(\widetilde{X}_i)) \\
    L_\lambda^{\PPI}(\theta)  & =L_n(\theta) -\lambda(L_n^f(\theta)-\widetilde{L}_N^f(\theta)).
\end{align*}
The PPI++ estimator for \(\theta^\star\) is defined to be
\begin{equation}
    \hat{\theta}^\PPI_\lambda = \argmin_\theta L_\lambda^\PPI(\theta).
\end{equation}
Note that if \(\ell_\theta(x,y)=\frac{1}{2}(y-\theta)^2\), then \(\theta^\star = \Expect[Y_i]\) and \(\hat\theta^\PPI_\lambda\) is equal to the expression in \eqref{eq:PPI++mean}. Specifically,

\begin{align*}
    \hat\theta^{\PPI}_\lambda & = \frac{1}{n}\sum_{i=1}^n Y_i - \lambda\left(\frac{1}{n}\sum_{i=1}^n f(X_i)-\frac{1}{N}\sum_{i=1}^N f(\widetilde{X}_i)\right) \\
                              & =\frac{\lambda}{N}\sum_{i=1}^N f(\widetilde{X}_i) + \frac{1}{n}\sum_{i=1}^n Y_i-\lambda f(X_i).                               \\
\end{align*}

\subsection{Data generating process for clustered data}\label{sec:appn-assumptions}

The assumptions for the data generating process are the same as in Section~\ref{sec:DGP}. Here we list them explicitly and define notation that will be used in subsequent sections.

\begin{assumption}[i.i.d. clusters]\label{as:iid}
    Clusters containing both features and labels are sampled i.i.d. from a cluster distribution. An individual cluster is represented as \((X_{g,i},Y_{g,i})_{i=1}^{m_g}\) where \(g=1,\ldots,G\) indexes the cluster, \(m_g\) is the size of cluster \(g\), and \(i=1,\ldots,m_g\) indexes the units within the cluster. As before, let \(\kappa = \Expect[m_g]\)
\end{assumption}

\begin{assumption}[Uniqueness of estimand]\label{as:unique}
    The estimand \(\theta^\star\) is the unique minimizer of the population loss
    \[
        \theta \mapsto \frac{1}{\kappa}\Expect\left[\sum_{i=1}^{m_g}\ell_\theta(X_{g,i},Y_{g,i})\right]=:L(\theta).
    \]
    And that the Hessian \(H_{\theta^\star} = \nabla^2_\theta L(\theta^\star)\) is positive definite.
\end{assumption}

\begin{assumption}[Bernoulli labeling within cluster]\label{as:labels}
    For each unit \((X_{g,i},Y_{g,i})\), there exists a label indicator \(\xi_{g,i} \in \{0,1\}\). If \(\xi_{g,i}=1\), then both \(X_{g,i}\) and \(Y_{g,i}\) are observed. If \(\xi_{g,i}=0\), then only \(X_{g,i}\) is observed. It is assumed that there exists a cluster level probability \(p_g\) such that \(p_g\) is independent of \((X_{g,i},Y_{g,i})_{i=1}^{m_g}\) and \((\xi_{g,i})_{i=1}^{m_g} \mid p_g\) are i.i.d Bernoulli with \(\prob(\xi_{g,i}=1\mid p_g) = p_g\). It is assumed that the probabilities \((p_g)_{g=1}^G\) are themselves i.i.d. with \(0 < \Expect[p_g]<1\). Define \(r = \Expect[p_g]/(1-\Expect[p_g])\) to be the expected ratio of labeled to unlabeled data. Each probability \(p_g\) is assumed to be independent of the predictors and outcomes \((X_{g,i}, Y_{g,i})_{i=1}^{m_g}\) as well as the cluster size \(m_g\).
\end{assumption}

\begin{assumption}[Convexity]\label{as:convex}
    There exists a subset \(\Lambda \subseteq \reals\) such that for all \(\lambda \in \Lambda\) and \(x_1,y_1\) and \(x_2\), the function
    \[
        \theta \mapsto L_n(\theta) - \lambda(L_n^f(\theta)-\widetilde{L}_N^f(\theta))
    \]
    is convex for all \((X_i,Y_i)_{i=1}^n\) and \((\widetilde{X}_i)_{i=1}^N\)
\end{assumption}
\begin{assumption}[Regularity conditions]\label{as:reg}
    There exists a compact convex set \(U\) which contains a neighborhood of \(\theta^\star\) on which \(\ell_\theta(x,y)\) and \(\ell_\theta(x, f(x))\) are twice continuously differentiable in \(\theta\).

    It is also assumed that for all \(\theta \in U\), the sums \(\sum_{i=1}^{m_g}\Vert \nabla_\theta \ell_\theta(X_{g,i},Y_{g,i})\Vert\) and \(\sum_{i=1}^{m_g}\Vert \nabla_\theta \ell_\theta(X_{g,i}, f(X_{g,i}))\Vert\) have finite variance.

    Finally, it assumed that \(\Expect[m_g^2] < \infty\) and \(\Expect[A_g] < \infty\) where
    \[
        A_g = \sum_{i=1}^{m_g}\sup_{\theta \in U}\left(\Vert \nabla^2_\theta\ell_\theta(X_{g,i},Y_{g,i})\Vert + \Vert\nabla^2_\theta\ell_\theta(X_{g,i},f(X_{g,i}))\Vert \right).
    \]
\end{assumption}

\begin{remark}
    For mean estimation, Assumption~\ref{as:convex} holds with \(\Lambda = \reals\). This is because if \(\ell_\theta(x,y) = \frac{1}{2}(\theta-y)^2\), then
    \begin{align*}
        \frac{d}{d\theta}\left(L_n(\theta) - \lambda(L_n^f(\theta)-\widetilde{L}_N^f(\theta))\right)&=\frac{1}{n}\sum_{i=1}^n (\theta - Y_i) + \lambda\left(\frac{1}{n}\sum_{i=1}^n f(X_i) - \frac{1}{N}\sum_{i=1}^N f(\widetilde{X}_i)\right),
    \end{align*}
    which is increasing in \(\theta\) and hence \(\theta \mapsto L_n(\theta) - \lambda(L_n^f(\theta)-\widetilde{L}_N^f(\theta))\) is convex.

    For generalized linear models with canonical link functions (GLMs), Assumption~\ref{as:convex} holds with \(\Lambda = [0,1]\) \citep{angelopoulos_ppi_2024}. 
\end{remark}

\subsection{Cluster robust standard errors}

Suppose that our data consists of clusters generated according to Section~\ref{sec:appn-assumptions}. We will again assume that the estimand \(\theta^\star\) is defined via a unit-level loss function \(\ell_\theta(x,y)\) by
\begin{equation}
    \theta^\star = \argmin_{\theta \in \reals^d} \frac{1}{\kappa}\Expect\left[\sum_{i=1}^{m_g}\ell_{\theta}(X_{g,i},Y_{g,i})\right].
\end{equation}
We also allow for \(\theta^\star\) to depend on unit-level or cluster-level weights such as 
\[
    \theta^\star = \argmin_{\theta \in \reals^d} \Expect\left[\frac{1}{m_g}\sum_{i=1}^{m_g}\ell_{\theta}(X_{g,i},Y_{g,i})\right].
\]
In a slight abuse of notation, we will continue to write \(\sum_{i=1}^{m_g}\ell_{\theta}(X_{g,i},Y_{g,i})\) for the weighted version. When weights are used, all assumptions in Section~\ref{sec:appn-assumptions} are assumed to hold for the weighted loss function.

In the clustered setting, the PPI++ estimator \citep{angelopoulos_ppi_2024} can be written as
\begin{equation}
    \hat\theta^\PPI_\lambda = \argmin_{\theta \in \reals^d}\frac{1}{n}\sum_{g=1}^G\sum_{i=1}^{m_g} \xi_{g,i}\ell_{\theta}(X_{g,i},Y_{g,i})-\lambda\left(\xi_{g,i}-\frac{n}{N}(1-\xi_{g,i})\right)\ell_{\theta}(X_{g,i},f(X_{g,i})).
\end{equation}

For \(g=1,\ldots,G\) define
\begin{align*}
    \varepsilon_g & = \sum_{i=1}^{m_g}\xi_{g,i}\nabla_\theta \ell_{\theta^\star}(X_{g,i}, Y_{g,i}),                               \\
    Z_g           & =\sum_{i=1}^{m_g}(\xi_{g,i} - r(1-\xi_{g,i}))(\nabla_\theta\ell_{\theta^\star}(X_{g,i}, f(X_{g,i})) - \mu_f),
\end{align*}
where \(\mu_f = \Expect\left[\sum_{i=1}^{m_g}\nabla_\theta\ell_{\theta^\star}(X_{g,i}, f(X_{g,i}))\right]/\Expect[m_g]\). Let \(\Sigma \in \reals^{2d \times 2d}\) be the population covariance matrix \(\cov(\varepsilon_g, Z_g)\). Also define
\[
    H_{\theta^\star} = \frac{1}{\kappa}\Expect\left[\sum_{i=1}^{m_g} \nabla^2_\theta\ell_{\theta^\star}(X_{g,i},Y_{g,i})\right] \in \reals^{d \times d},
\]
where \(\kappa = \Expect[m_g]\). Define the sample analogs of \(\varepsilon_g, \mu_f, Z_g,\Sigma\) and \(H_{\theta^\star}\) by
\begin{align*}
    \hat\varepsilon_g & = \sum_{i=1}^{m_g}\xi_{g,i}\nabla_\theta \ell_{\hat\theta^\PPI_\lambda}(X_{g,i}, Y_{g,i})                                                                                                                             \\
    \hat\mu_f         & =\frac{1}{n+N}\sum_{g=1}^G\sum_{i=1}^{m_g}\nabla_{\theta} \ell_{\hat\theta^\PPI_\lambda}(X_{g,i},f(X_{g,i}))                                                                                                          \\
    \hat Z_g          & =\sum_{i=1}^{m_g}\left(\xi_{g,i}-\frac{n}{N}(1-\xi_{g,i})\right)(\nabla_\theta\ell_{\hat\theta^\PPI_\lambda}(X_{g,i}, f(X_{g,i}))-\hat{\mu}_f )\\
    \hat{\Sigma}      & =\frac{1}{G}\sum_{g=1}^G \begin{bmatrix}
                                                     \hat{\varepsilon}_g\hat{\varepsilon}_g^\top & \hat{\varepsilon}_g\hat{Z}_g^\top \\
                                                     \hat{Z}_g\hat{\varepsilon}_g^\top           & \hat{Z}_g\hat{Z}_g^\top
                                                 \end{bmatrix}                                                                                                              \\
    \hat{H}           & =\frac{1}{n}\sum_{g=1}^G\sum_{i=1}^{m_g} \xi_{g,i}\nabla^2_\theta\ell_{\hat\theta^\PPI_\lambda}(X_{g,i},Y_{g,i})
\end{align*}
Note that \(\hat\varepsilon_g, \hat\mu_f, \hat Z_g,\hat\Sigma\) and \(\hat H\) all implicitly depend on \(\lambda\).

The following result gives a central limit theorem for \(\hat\theta^\PPI_\lambda\) that is a generalization of Theorem~\ref{thrm:cluster-clt}.

\begin{theorem}[Central limit theorem for the general PPI++ estimator]
    \label{thrm:general-clt}
    Suppose that the assumptions in Section~\ref{sec:appn-assumptions} hold and that \(\lambda \in \Lambda\). Then, as \(G \to \infty\),
    \begin{align*}
        \sqrt{n}(\hat\theta^\PPI_\lambda-\theta^\star) \stackrel{d}{\to} \mathcal{N}(0,V_\lambda),
    \end{align*}
    where \(V_\lambda \in \reals^{d \times d}\) is given by
    \begin{align*}
        V_\lambda & =\frac{1+r}{r\kappa}H_{\theta^\star}^{-1}\left(\Sigma_{1:d,1:d}-\lambda(\Sigma_{1:d,(d+1):(2d)}+\Sigma_{(d+1):(2d),1:d})+\lambda^2 \Sigma_{(d+1):(2d),(d+1):(2d)}\right)H_{\theta^\star}^{-1}  \\
                  & =\frac{1+r}{r\kappa}H_{\theta^\star}^{-1}\left(\cov(\varepsilon_g,\varepsilon_g)-\lambda(\cov(\varepsilon_g,Z_g)+\cov(Z_g,\varepsilon_g))+\lambda^2 \cov(Z_g,Z_g)\right)H_{\theta^\star}^{-1}.
    \end{align*}
    Furthermore, if \(V_\lambda\) is positive definite and
    \[
        \hat{V}_\lambda = \frac{G}{n}\hat{H}^{-1}\left(\hat\Sigma_{1:d,1:d}-\lambda(\hat\Sigma_{1:d,(d+1):(2d)}+\hat\Sigma_{(d+1):(2d),1:d})+\lambda^2 \hat\Sigma_{(d+1):(2d),(d+1):(2d)}\right)\hat H^{-1},
    \]
    then \(\hat{V}_\lambda\) converges to \(V_\lambda\) in probability and thus
    \begin{align*}
        \sqrt{n}\hat{V}_\lambda^{-1/2}\left(\hat\theta^\PPI_\lambda - \theta^\star\right) \stackrel{d}{\to} \mathcal{N}(0,I_d),
    \end{align*}
    where \(I_d\) is the \(d \times d\) identity matrix.
\end{theorem}
Theorem~\ref{thrm:general-clt} is proved in Appendix~\ref{sec:proofs}. Cluster robust confidence intervals for \(\theta^\star_{j}\) are created by computing \(\hat\theta^\PPI_{\lambda}\) and \(\hat{V}_\lambda\) and returning
\[
    \mathcal{C}_\alpha^\PPI = \left([\hat\theta^\PPI_{\lambda}]_{j}\pm z_{1-\alpha/2}\sqrt{[\hat{V}_{\lambda}]_{j,j}/n}\right),
\]
where \([\hat\theta^\PPI_{\lambda}]_{j}\) denotes the \(j\)th entry of \(\hat\theta^\PPI_\lambda\) and \([\hat{V}_{\lambda}]_{j,j}\) denotes the \((j,j)\)th entry of \(\hat{V}_\lambda\).

Note that by taking \(\ell_\theta(x,y) = \frac{1}{2}(y-\theta)^2\), Theorem~\ref{thrm:general-clt} implies Theorem~\ref{thrm:cluster-clt} and Remark~\ref{rem:CI} from the main text.

\subsection{Power tuning}

We consider two different procedures for choosing \(\lambda\). The first minimizes the trace \(\Tr(V_\lambda)\) and the second minimizes \([V_\lambda]_{j,j}\) for a coordinate of interest \(j\).

For trace minimization we have

\begin{align*}
    \lambda^\star_{\Tr} & := \frac{\Tr\left(H_{\theta^\star}^{-1}\left(\Sigma_{1:d,(d+1):(2d)}+\Sigma_{(d+1):(2d),1:d}\right)H_{\theta^\star}^{-1}\right)}{2\Tr\left(H_{\theta^\star}^{-1}\Sigma_{(d+1):(2d),(d+1):(2d)}H_{\theta^\star}^{-1}\right)} \\
                        & =\frac{\Tr\left(H_{\theta^\star}^{-1}\Sigma_{1:d,(d+1):(2d)}H_{\theta^\star}^{-1}\right)}{\Tr\left(H_{\theta^\star}^{-1}\Sigma_{(d+1):(2d),(d+1):(2d)}H_{\theta^\star}^{-1}\right)},
\end{align*}
where the second line used the fact that \((H_{\theta^\star}^{-1}\Sigma_{1:d,(d+1):(2d)}H_{\theta^\star}^{-1})^\top = H_{\theta^\star}^{-1}\Sigma_{(d+1):(2d),1:d}H_{\theta^\star}^{-1}\) and thus these two matrices have the same trace. The optimal parameter \(\hat\lambda^\star_{\Tr}\) can be estimated via
\begin{align*}
    \hat\lambda_{\Tr} & := \frac{\Tr\left(\hat H^{-1}\hat\Sigma_{1:d,(d+1):(2d)}\hat H^{-1}\right)}{\Tr\left(\hat H^{-1}\hat\Sigma_{(d+1):(2d),(d+1):(2d)}\hat H^{-1}\right)},
\end{align*}
where \(\widehat{\Sigma}\) is first computed using \(\lambda=1\) as in \citet{angelopoulos_ppi_2024}.

For GLMs, the parameter \(\hat\lambda_{\Tr}\) is clipped to be between \(0\) and \(1\) to ensure that the loss function remains convex.

The value of \(\lambda\) that minimizes \([V_\lambda]_{j,j}\) is given by
\begin{align*}
    \lambda^\star_{j} & := \frac{\left[H_{\theta^\star}^{-1}\left(\Sigma_{1:d,(d+1):(2d)}+\Sigma_{(d+1):(2d),1:d}\right)H_{\theta^\star}^{-1}\right]_{j,j}}{2\left[H_{\theta^\star}^{-1}\Sigma_{(d+1):(2d),(d+1):(2d)}H_{\theta^\star}^{-1}\right]_{j,j}} \\
                      & =\frac{\left[H_{\theta^\star}^{-1}\Sigma_{1:d,(d+1):(2d)}H_{\theta^\star}^{-1}\right]_{j,j}}{\left[H_{\theta^\star}^{-1}\Sigma_{(d+1):(2d),(d+1):(2d)}H_{\theta^\star}^{-1}\right]_{j,j}},
\end{align*}
where the second line again uses \((H_{\theta^\star}^{-1}\Sigma_{1:d,(d+1):(2d)}H_{\theta^\star}^{-1})^\top = H_{\theta^\star}^{-1}\Sigma_{(d+1):(2d),1:d}H_{\theta^\star}^{-1}\) and thus these two matrices have the same diagonal entries. The optimal parameter \(\hat\lambda^\star_{j}\) can be estimated via
\begin{align*}
    \hat\lambda_{j} & := \frac{\left[\hat H^{-1}\hat\Sigma_{1:d,(d+1):(2d)}\hat H^{-1}\right]_{j,j}}{\left[\hat H^{-1}\hat\Sigma_{(d+1):(2d),(d+1):(2d)}\hat H^{-1}\right]_{j,j}},
\end{align*}
where \(\widehat{\Sigma}\) is first computed using \(\lambda=1\) as in \citet{angelopoulos_ppi_2024}.

For GLMs, the parameter \(\hat\lambda_{j}\) is clipped to be between \(0\) and \(1\) to ensure that the loss function remains convex.

The following theorem gives a central limit theorem for \(\hat{\theta}_{\hat\lambda}^{\PPI}\) where \(\hat\lambda\) is any estimator that converges in probability to a fixed \(\lambda\).

\begin{theorem}[Central limit theorem for power-tuned estimator]
    \label{thrm:power-tuned-clt}
    Suppose that the assumptions in Section~\ref{sec:appn-assumptions} hold and let \(V_\lambda\) and \(\hat{V}_{\lambda}\) be as in Theorem~\ref{thrm:general-clt}. If \(\hat\lambda\) converges to \(\lambda\) in probability and both \(\hat\lambda\) and \(\lambda\) are in \(\Lambda\), then
    \begin{align*}
        \sqrt{n}(\hat\theta^\PPI_{\hat\lambda}-\theta^\star) \stackrel{d}{\to} \mathcal{N}(0,V_\lambda),
    \end{align*}
    and
    \begin{align*}
        \sqrt{n}\hat{V}_{\hat\lambda}^{-1/2}\left(\hat\theta^\PPI_{\hat\lambda} - \theta^\star\right) \stackrel{d}{\to} \mathcal{N}(0,I_d),
    \end{align*}
    where \(I_d\) is the \(d \times d\) identity matrix.
\end{theorem}
Thus, asymptotically valid confidence intervals can be based on \(\hat\theta_{\hat\lambda}\) and \(\hat{V}_{\hat\lambda}\). As before, Theorem~\ref{thrm:power-tuned-clt} implies equation \eqref{eq:power-tuned-clt} which states a central limit theorem for the power tuned mean estimator.

\subsection{The PPI-correlation}

Let \(\sigma^2_{\lambda}\) be the asymptotic variance of \([\hat\theta^\PPI_{\lambda}]_{j}\). If \(\lambda^\star = \lambda_j^\star \in [0,1]\), then we have
\begin{align*}
    \sigma^2_{\lambda^\star}  = [V_{\lambda^\star}]_{j,j}
     & =\frac{1+r}{r\kappa}\left(\left[H^{-1}_{\theta^\star}\Sigma_{1:d,1:d}H^{-1}_{\theta^\star}\right]_{j,j} - \frac{\left[H^{-1}_{\theta^\star}\Sigma_{1:d,(d+1):(2d)}H^{-1}_{\theta^\star}\right]_{j,j}^2}{\left[H^{-1}_{\theta^\star}\Sigma_{(d+1):(2d),(d+1):(2d)}H^{-1}_{\theta^\star}\right]_{j,j}} \right)                                                                                          \\
     & =\frac{1+r}{r\kappa}\left[H^{-1}_{\theta^\star}\Sigma_{1:d,1:d}H^{-1}_{\theta^\star}\right]_{j,j}\left(1- \frac{\left[H^{-1}_{\theta^\star}\Sigma_{1:d,(d+1):(2d)}H^{-1}_{\theta^\star}\right]_{j,j}^2}{\left[H^{-1}_{\theta^\star}\Sigma_{1:d,1:d}H^{-1}_{\theta^\star}\right]_{j,j}\left[H^{-1}_{\theta^\star}\Sigma_{(d+1):(2d),(d+1):(2d)}H^{-1}_{\theta^\star}\right]_{j,j}}\right) \\
     & =\sigma^2_0(1-\rho^2),
\end{align*}
where \(\sigma^2_0\) is the asymptotic variance of the labels-only estimator \(\hat\theta^\PPI_0\) and
\[
    \rho =\frac{\left[H^{-1}_{\theta^\star}\Sigma_{1:d,(d+1):2d}H^{-1}_{\theta^\star}\right]_{j,j}}{\sqrt{\left[H^{-1}_{\theta^\star}\Sigma_{1:d,1:d}H^{-1}_{\theta^\star}\right]_{j,j}\left[H^{-1}_{\theta^\star}\Sigma_{(d+1):(2d),(d+1):(2d)}H^{-1}_{\theta^\star}\right]_{j,j}}}.
\]
See Section~\ref{sec:appn-var} for more explicit formulas for \(\rho\) in the case of mean estimation.

\section{Proofs}\label{sec:proofs}

\subsection{Proof of Theorem~\ref{thrm:general-clt}}

Through-out, let \(\hat\theta\) denote \(\hat\theta_\lambda^\PPI\) and 
\[\Sigma_\lambda = \cov(\varepsilon_g - \lambda Z_g) = \Sigma_{1:d,1:d}-\lambda(\Sigma_{1:d,(d+1):2d}+\Sigma_{(d+1):2d,1:d})+\lambda^2 \Sigma_{(d+1):2d, (d+1):2d}.\]
Also let 
\begin{align*}
    L^\PPI_G(\theta) = \frac{1}{n}\sum_{g=1}^G \sum_{i=1}^{m_g} \xi_{g,i}\ell_\theta(X_{g,i}, Y_{g,i}) - \lambda\left(\xi_{g,i} - \frac{n}{N}(1-\xi_{g,i})\right)\ell_\theta(X_{g,i},f(X_{g,i})),
\end{align*}
and 
\begin{align*}
    L(\theta) = \frac{1}{\kappa}\Expect\left[\sum_{i=1}^{m_g}\ell_\theta(X_{g,i},Y_{g,i})\right].
\end{align*}

Theorem~2.7 of \citet{newey1994large} and the convexity of \(L^\PPI_G\) imply that the PPI estimator is consistent.
\begin{proposition}
    The estimate \(\hat\theta\) is consistent for \(\theta^\star\). That is \(\hat\theta \stackrel{p}{\to}\theta^\star\) as \(G \to \infty\).
\end{proposition}
\begin{proof}
    To apply Theorem~2.7 of \citet{newey1994large}, note that we have assumed that \(\theta^\star\) is the unique minimizer of \(L(\theta)\). Furthermore, the independence between \((\xi_{g,i})_{i=1}^{m_g}\) and \((X_{g,i},Y_{g,i})_{i=1}^{m_g}\) implies that 
    \[ 
        L^\PPI_G(\theta) \stackrel{p}{\to} L(\theta)
    \]
    for all \(\theta \in U\) where \(U\) is the compact convex set containing a neighborhood of \(\theta^\star\) from Assumption~\ref{as:reg}. Finally, the assumption that \(\lambda \in \Lambda\) implies that \(L^\PPI_G\) is convex in \(\theta\) as required.
\end{proof}

The next lemma gives a uniform law of large numbers for the Hessian \(\nabla^2_\theta L^\PPI_G\).

\begin{lemma}\label{lem:ULLN}
   Let \(U\) be the set from Assumption~\ref{as:reg}, then 
   \[
    \sup_{\theta \in U} \left\Vert \nabla^2_\theta L^\PPI_G(\theta) - \nabla^2_\theta L(\theta)\right\Vert \stackrel{p}{\to} 0.
   \]
\end{lemma}
\begin{proof}
    By independence of \((\xi_{g,i})_{i=1}^{m_g}\) and \((X_{g,i},Y_{g,i})_{i=1}^{m_g}\), it follows that
    \begin{align*}
        \Expect\left[\frac{1+r}{r\kappa}\nabla^2_\theta\sum_{i=1}^{m_g}\xi_{i,g} \ell_\theta(X_{g,i},Y_{g,i})\right]& = \frac{(1+r)\Expect[p_g]}{r\kappa}\Expect\left[\nabla^2_\theta\sum_{i=1}^{m_g}\ell_\theta(X_{g,i},Y_{g,i})\right]\\
        &=\nabla^2_\theta L(\theta).
    \end{align*}
    And likewise
    \begin{align*}
        \Expect\left[\frac{1+r}{r\kappa}\nabla^2_\theta\sum_{i=1}^{m_g}\xi_{i,g}\ell_\theta(X_{g,i},f(X_{g,i})) \right]&=\nabla^2_\theta L^f(\theta), \text{ and}\\
        \Expect\left[\frac{1+r}{\kappa}\nabla^2_\theta\sum_{i=1}^{m_g}(1-\xi_{i,g})\ell_\theta(X_{g,i},f(X_{g,i})) \right]&=\nabla^2_\theta L^f(\theta),
    \end{align*}
    where
    \[
        L^f(\theta) = \frac{1}{\kappa}\Expect\left[\sum_{i=1}^{m_g}\ell_\theta(X_{g,i},f(X_{g,i}))\right].
    \]
    By Assumption~\ref{as:reg} and Lemma~2.4 of \citet{newey1994large},
    \begin{align*}
        \sup_{\theta \in U}\left\Vert \frac{1+r}{r\kappa G}\nabla_\theta^2\sum_{g=1}^G\sum_{i=1}^{m_g}\xi_{g,i}\ell_\theta(X_{g,i},Y_{g,i}) - \nabla^2_\theta L(\theta)\right\Vert \stackrel{p}{\to} 0,
    \end{align*}
    and
    \begin{align*}
        \sup_{\theta \in U}\left\Vert \frac{1+r}{r\kappa G}\nabla_\theta^2\sum_{g=1}^G\sum_{i=1}^{m_g}\xi_{g,i}\ell_\theta(X_{g,i},f(X_{g,i})) - \frac{1+r}{\kappa G}\nabla_\theta^2\sum_{g=1}^G\sum_{i=1}^{m_g}(1-\xi_{g,i})\ell_\theta(X_{g,i},f(X_{g,i}))\right\Vert \stackrel{p}{\to} 0.
    \end{align*}
    Furthermore, we have that \(\frac{n}{G} \stackrel{p}{\to} \frac{r \kappa}{1+r}\) and \(\frac{N}{G} \stackrel{p}{\to} \frac{\kappa}{1+r}\). This, together with Assumption~\ref{as:reg}, implies that
    \begin{align*}
        \sup_{\theta \in U}\left\Vert\frac{1}{n} \nabla_\theta^2\sum_{g=1}^G\sum_{i=1}^{m_g}\xi_{g,i}\ell_\theta(X_{g,i},Y_{g,i}) - \frac{1+r}{r\kappa G}\sum_{g=1}^G\sum_{i=1}^{m_g}\xi_{g,i}\nabla_\theta^2\ell_\theta(X_{g,i},Y_{g,i}) \right\Vert \stackrel{p}{\to} 0.
    \end{align*}
    And likewise 
    \begin{align*}
        \sup_{\theta \in U}\left\Vert \frac{1}{n}\nabla_\theta^2\sum_{g=1}^G\sum_{i=1}^{m_g}\xi_{g,i}\ell_\theta(X_{g,i},f(X_{g,i})) -  \frac{1+r}{r\kappa G}\nabla_\theta^2\sum_{g=1}^G\sum_{i=1}^{m_g}\xi_{g,i}\ell_\theta(X_{g,i},f(X_{g,i}))\right\Vert &\stackrel{p}{\to} 0, \text{ and}\\
        \sup_{\theta \in U}\left\Vert \frac{1}{N}\nabla_\theta^2\sum_{g=1}^G\sum_{i=1}^{m_g}(1-\xi_{g,i})\ell_\theta(X_{g,i},f(X_{g,i})) -  \frac{1+r}{\kappa G}\nabla_\theta^2\sum_{g=1}^G\sum_{i=1}^{m_g}(1-\xi_{g,i})\ell_\theta(X_{g,i},f(X_{g,i}))\right\Vert &\stackrel{p}{\to} 0.
    \end{align*}
    The triangle inequality then gives
    \begin{align*}
         \sup_{\theta \in U} \left\Vert \nabla^2_\theta L^\PPI_G(\theta) - \nabla^2_\theta L(\theta)\right\Vert \stackrel{p}{\to} 0,
    \end{align*}
    as required.
\end{proof}
The next lemma is a simple Taylor's approximation for the gradient \(\nabla_\theta L_G^\PPI(\theta^\star)\). 
\begin{lemma}\label{lem:taylor}
    If \(\hat\theta \in U\) and
    \(\bar{H} = \int_0^1 \nabla^2_\theta L_G^\PPI(\hat\theta + t(\theta^\star - \hat\theta))dt\)    then
    \begin{align*}
        \nabla_\theta  L^\PPI_G(\theta^\star) = \bar{H}(\theta^\star - \hat\theta)
    \end{align*}
    and \(\Vert \bar{H}-H\Vert \stackrel{p}{\to} 0\).
\end{lemma}
\begin{proof}
    Define \(h : [0,1] \to \reals^d\) by \(h(t) = \nabla_\theta L^\PPI_G(\hat\theta + t(\theta^\star - \hat\theta))\). The function \(h\) satisfies \(h(0)=0\), \(h(1) = \nabla_\theta L_G^\PPI(\theta^\star)\) and
    \[
        h'(t) = \nabla_\theta^2 L^\PPI_G(\hat\theta + t(\theta^\star - \hat\theta))(\theta^\star - \hat\theta).
    \]
    The fundamental theorem of calculus gives
    \begin{align*}
        \nabla_\theta L^\PPI_G(\theta^\star)&=h(0)+\int_0^1 h'(t)dt\\
        &=\bar{H}(\theta^\star - \hat \theta). 
    \end{align*}
    Furthermore, if \(\hat\theta \in U\), then the line segment connecting \(\hat\theta\) and \(\theta^\star\) is also in \(U\) and thus \(\bar{H} \to H_{\theta^\star}\) by the uniform law of large numbers in Lemma~\ref{lem:ULLN}.
\end{proof}
Lemma~\ref{lem:gradient} relates the gradient \(\nabla_\theta L_G^\PPI(\theta^\star)\) to \(\varepsilon_g\) and \(Z_g\).
\begin{lemma}\label{lem:gradient}
    The gradient \(\nabla_\theta L_G^\PPI(\theta^\star)\) satisfies
    \begin{align*}
        \sqrt{n}\nabla_\theta L_G^\PPI(\theta^\star)&=\frac{1}{\sqrt{n}}\sum_{g=1}^G \left(\varepsilon_g - \lambda Z_g\right) + O_p(G^{-1/2}).
    \end{align*}
\end{lemma}
\begin{proof}
    By definition,
    \begin{align*}
        \frac{1}{\sqrt{n}}\sum_{g=1}^G\sum_{i=1}^{m_g} \xi_{g,i}\nabla_{\theta} \ell_{\theta^\star}(X_{g,i},Y_{g,i}) &=\frac{1}{\sqrt{n}}\sum_{g=1}^G \varepsilon_g. 
    \end{align*}
    Furthermore, since \(\sum_{g=1}^G \sum_{i=1}^{m_g}\left(\xi_{g,i}-\frac{n}{N}(1-\xi_{g,i})\right)=0\), we have
    \begin{align*}
        &\frac{1}{\sqrt{n}}\sum_{g=1}^G\sum_{i=1}^{m_g}\left(\xi_{g,i}-\frac{n}{N}(1-\xi_{g,i})\right)\nabla_{\theta}\ell_{\theta^\star}(X_{g,i}, f(X_{g,i}))\\
         &= \frac{1}{\sqrt{n}}\sum_{g=1}^G\sum_{i=1}^{m_g}\left(\xi_{g,i}-\frac{n}{N}(1-\xi_{g,i})\right)(\nabla_{\theta}\ell_{\theta^\star}(X_{g,i}, f(X_{g,i}))-\mu_f),
    \end{align*}
    where \(\mu_f = \frac{1}{\kappa}\Expect\left[\sum_{i=1}^{m_g}\nabla_\theta \ell_{\theta^\star}(X_{g,i},f(X_{g,i}))\right]\). Next, replace \(n/N\) with the limiting value \(r\):
    \begin{align*}
        &\frac{1}{\sqrt{n}}\sum_{g=1}^G\sum_{i=1}^{m_g}\left(\xi_{g,i}-\frac{n}{N}(1-\xi_{g,i})\right)\nabla_{\theta}\ell_{\theta^\star}(X_{g,i}, f(X_{g,i}))\\
        &=\frac{1}{\sqrt{n}}\sum_{g=1}^GZ_g + \frac{1}{\sqrt{n}}\sum_{g=1}^G\sum_{i=1}^{m_g}\left(r-\frac{n}{N}\right)(1-\xi_{g,i})(\nabla_{\theta}\ell_{\theta^\star}(X_{g,i}, f(X_{g,i}))-\mu_f)\\
        &=\frac{1}{\sqrt{n}}\sum_{g=1}^GZ_g + \left(r-\frac{n}{N}\right)\frac{1}{\sqrt{n}}\sum_{g=1}^G\sum_{i=1}^{m_g}(1-\xi_{g,i})(\nabla_{\theta}\ell_{\theta^\star}(X_{g,i}, f(X_{g,i}))-\mu_f).
    \end{align*}
    By the central limit theorem, \(r - \frac{n}{N}=O_p(G^{-1/2})\). Furthermore, \((\xi_{g,i})_{i=1}^{m_g}\) are independent of \((X_{g,i},Y_{g,i})_{i=1}^{m_g}\) and \(\sum_{i=1}^{m_g}(\nabla_\theta \ell_{\theta^\star}(X_{g,i},f(X_{g,i}))-\mu_f)\) has mean zero. Since \( \frac{G}{n} = O_p(1)\), these results give
    \[
        \frac{1}{\sqrt{n}}\sum_{g=1}^G\sum_{i=1}^{m_g}(1-\xi_{g,i})(\nabla_{\theta}\ell_{\theta^\star}(X_{g,i}, f(X_{g,i}))-\mu_f) = O_p(1).
    \]
    And thus,
    \begin{align*}
        &\frac{1}{\sqrt{n}}\sum_{g=1}^G\sum_{i=1}^{m_g}\left(\xi_{g,i}-\frac{n}{N}(1-\xi_{g,i})\right)\nabla_{\theta}\ell_{\theta^\star}(X_{g,i}, f(X_{g,i}))\\
        &=\frac{1}{\sqrt{n}}\sum_{g=1}^GZ_g+O_p(G^{-1/2}),
    \end{align*}
    which implies that
    \[
    \sqrt{n}\nabla_\theta L_G^\PPI(\theta^\star) = \frac{1}{\sqrt{n}}\sum_{g=1}^G(\varepsilon_g - \lambda Z_g)+O_p(G^{-1/2}).
    \]
\end{proof}
We are now ready to prove a CLT for \(\nabla_\theta L^\PPI_G(\theta^\star)\) and hence for \(\hat\theta - \theta^\star\). Recall that \(\Sigma_\lambda = \cov(\varepsilon_g - \lambda Z_g)\).
\begin{lemma}
    As \(G \to \infty\)
    \begin{align}\label{eq:CLT1}
        \frac{1}{\sqrt{G}}\sum_{g=1}^{G}(\varepsilon_g - \lambda Z_g) \to \mathcal{N}(0, \Sigma_\lambda),
    \end{align}
    and hence
    \begin{align}\label{eq:CLT2}
        \sqrt{n}\nabla_\theta L^\PPI_G(\theta^\star) \to \mathcal{N}\left(0, \frac{1+r}{r\kappa}\Sigma_\lambda\right),
    \end{align}
    and
    \begin{align}\label{eq:CLT3}
        \sqrt{n}(\hat\theta - \theta^\star) \to \mathcal{N}\left(0, \frac{1+r}{r\kappa}H^{-1}_{\theta^\star}\Sigma_\lambda H^{-1}_{\theta^\star}\right).
    \end{align}
\end{lemma}
\begin{proof}
    Equation \eqref{eq:CLT1}, follows from the Assumption~\ref{as:reg} and the standard multivariate central limit theorem. The second central limit theorem in \eqref{eq:CLT2} follows from Lemma~\ref{lem:gradient} and the fact that \(\frac{G}{n} \to \frac{1+r}{r\kappa}\).
    
    Finally, Lemma~\ref{lem:taylor} gives that
    \begin{align*}
        \nabla_\theta L_G^\PPI(\theta^\star)&=-(1+o_p(1))H_{\theta^\star}(\hat\theta - \theta^\star),
    \end{align*}
    Thus, \eqref{eq:CLT3} follows from \eqref{eq:CLT2} and Assumption~\ref{as:unique} which states that \(H_{\theta^\star}\) is invertible. The central limit theorem in Theorem~\ref{thrm:general-clt} is proved.
\end{proof}
It remains to show that the variance estimator \(\hat{\Sigma}\) is consistent for \(\Sigma\), but this follows from the moment assumptions in Assumption~\ref{as:reg} and the weak law of large numbers.

\subsection{Proof of Theorem~\ref{thrm:power-tuned-clt}}

Theorem~\ref{thrm:power-tuned-clt} follows from the lemmas in the previous section and the following estimate:
\begin{lemma}\label{lem:grads-close}
    Suppose that \(\hat{\lambda},\lambda \in \Lambda\) and \(\hat{\lambda} \stackrel{p}{\to} \lambda\). Let
    \[
        L^\PPI_{G,\lambda}(\theta) = \frac{1}{n}\sum_{g=1}^G \sum_{i=1}^{m_g} \xi_{g,i}\ell_\theta(X_{g,i}, Y_{g,i}) - \lambda\left(\xi_{g,i} - \frac{n}{N}(1-\xi_{g,i})\right)\ell_\theta(X_{g,i},f(X_{g,i})).
    \] 
    Then
    \begin{align*}
         \Vert \nabla_\theta L^\PPI_{G,\hat\lambda}(\theta^\star) - \nabla_\theta L^\PPI_{G,\lambda}(\theta^\star)\Vert = o_p(G^{-1/2}).
    \end{align*}
\end{lemma}
\begin{proof}
    Note that 
    \begin{align*}
    \nabla_\theta L^\PPI_{G,\hat\lambda}(\theta^\star) - \nabla_\theta L^\PPI_{G,\lambda}(\theta^\star)
    &=-(\hat{\lambda} - \lambda)\frac{1}{n}\sum_{g=1}^G\sum_{i=1}^{m_g} \left(\xi_{g,i} - \frac{n}{N}(1-\xi_{g,i})\right)\nabla_\theta \ell_\theta(X_{g,i},f(X_{g,i}))\\
    &=-\frac{1}{n}(\hat{\lambda}-\lambda)O_p(\sqrt{G})\\
    &=o_p(G^{-1/2}),
    \end{align*}
    where the last line follows from the fact that \(\hat{\lambda} \stackrel{p}{\to} \lambda\) and \(\frac{n}{G} \stackrel{p}{\to} \frac{\kappa r}{1+r}\).
\end{proof}
The central limit theorem then follows by the same argument outlined in the previous section.

\section{Additional formulas for the PPI variance}\label{sec:appn-var}

Define
\begin{align*}
    \mu_p &= \Expect[p_g] = \frac{r}{1+r}, \text{ and }\\
    \eta_p &=\frac{\var(p_g)}{\mu_p(1-\mu_p)}.
\end{align*}
The quantity \(\eta_p\) is exactly the correlation between two label indicators \(\xi_{g,i}\) and \(\xi_{g,j}\) for different units within the same cluster. The value \(\mu_p\) and \(\eta_p\) are the main design parameters for Cluster-Robust PPI++. The mean \(\mu_p\) determines the expected number of labeled units and the intraclass correlation \(\eta_p\) determines whether the labels are evenly distributed across clusters (\(\eta_p \approx 0\)) or concentrated in a fraction of clusters (\(\eta_p \approx 1\)). The number of sampled cluster \(G\) is also a design parameter.

The variance of the PPI++ estimator also depends on a number of \emph{non-design moments} that capture the variability and correlation in the outcomes and predictions. In the below notation, \(s\) refers to moments that involve the sums over the \emph{same} unit whereas \(d\) refers to moments that sum over \emph{different} units within the same cluster.
\begin{align*}
    V_{Y,s}&=\Expect\left[\sum_{i=1}^{m_g}(Y_{g,i}-\theta^\star)^2\right],\\
    V_{Y,d}&=\Expect\left[\sum_{i\neq j} (Y_{g,i}-\theta^\star)(Y_{g,j}-\theta^\star)\right],\\
    V_{f,s}&=\Expect\left[\sum_{i=1}^{m_g}(f(X_{g,i})-\theta^f)^2\right],\\
    V_{f,d}&=\Expect\left[\sum_{i\neq j} (f(X_{g,i})-\theta^f)(f(X_{g,j})-\theta^f)\right],\\
    V_{Y,f,s}&=\Expect\left[\sum_{i=1}^{m_g}(Y_{g,i}-\theta^\star)(f(X_{g,i})-\theta^f)\right], \text{ and}\\
    V_{Y,f,d}&=\Expect\left[\sum_{i\neq j} (Y_{g,i}-\theta^\star)(f(X_{g,j})-\theta^f)\right].
\end{align*}

The asymptotic variance of the power-tuned PPI++ estimator for the mean is
\begin{align}\label{eq:design-var}
    V^\PPI &= \frac{1}{\kappa^2 G\mu_p}\left(V_{\varepsilon}- \frac{V_{\varepsilon,Z}^2}{V_Z}\right),
\end{align}
where
\begin{align*}
    V_{\varepsilon} &=V_{Y,s}+\left(\eta_p(1-\mu_p) + \mu_p\right)V_{Y,d},\\
    V_{Z} &=\frac{1}{1-\mu_p}\left(V_{f,s}+\eta_pV_{f,d}\right), \text{ and}\\
    V_{\varepsilon,Z} &=V_{Y,f,s}+\eta_pV_{Y,f,d}.
\end{align*}

And the correlation between \(\varepsilon_g\) and \(Z_g\) can be written as
\begin{align}\label{eq:design-corr}
    \rho = \frac{V_{\varepsilon,Z}}{\sqrt{V_\varepsilon V_Z}}.
\end{align}
Equations~\ref{eq:design-var} and \ref{eq:design-corr} show how the variance and precision gain of the Cluster-Robust PPI++ estimator depends on both the design parameters \(\mu_p\) and \(\eta_p\) and on the non-design moments \(V\). If the non-design moments are known, then equation \eqref{eq:design-var} can be used to compare different labeling strategies and \eqref{eq:design-corr} can be used to assess the benefit of using PPI.

Since the non-design moments are typically not known, the following equations provide in-sample estimators. These estimators could be computed on data from a pilot study that satisfies the assumptions of Cluster-Robust PPI++ as outlined in Sections~\ref{sec:DGP} and \ref{sec:appn-assumptions}. The estimators are 
\begin{align*}
    \hat{V}_{Y,s}&=\frac{1}{G}\sum_{g=1}^G \frac{1}{\hat{r}_1}\sum_{i=1}^{m_g}\xi_{g,i}(Y_{g,i}-\hat{\theta}^\PPI_{\hat{\lambda}})^2,\\
    \hat{V}_{Y,d}&=\frac{1}{G}\sum_{g=1}^G\frac{1}{\hat{r}_2}\sum_{i\neq j}\xi_{g,i}\xi_{g,j}(Y_{g,i}-\hat{\theta}^\PPI_{\hat{\lambda}})(Y_{g,j}-\hat{\theta}^\PPI_{\hat{\lambda}}),\\
    \hat{V}_{f,s}&=\frac{1}{G}\sum_{g=1}^G\sum_{i=1}^{m_g}(f(X_{g,i})-\hat{\theta}^f)^2, \\
    \hat{V}_{f,d}&=\frac{1}{G}\sum_{g=1}^G\sum_{i\neq j}(f(X_{g,i})-\hat\theta^f)(f(X_{g,j})-\hat\theta^f),\\
    \hat{V}_{Y,f,s}&=\frac{1}{G}\sum_{g=1}^G\frac{1}{\hat{r}_1}\sum_{i=1}^{m_g}\xi_{g,i}(Y_{g,i}-\hat\theta^\PPI_{\hat\lambda})(f(X_{g,i})-\hat\theta^f), \text{ and}\\
    \hat{V}_{Y,f,d}&=\frac{1}{G}\sum_{g=1}^G\frac{1}{\hat{r}_1}\sum_{i\neq j}\xi_{g,i}(Y_{g,i}-\hat\theta^\PPI_{\hat\lambda})(f(X_{g,j})-\hat\theta^f),
\end{align*}
where 
\begin{align*}
    \hat{\theta}^f &= \frac{1}{n+N}\sum_{g=1}^G \sum_{i=1}^{m_g}f(X_{g,i}),\\
    \hat{r}_1 &=\frac{n}{n+N}, \text{ and}\\
    \hat{r}_2 &=\frac{\sum_{g=1}^G \sum_{i \neq j}\xi_{g,i}\xi_{g,j}}{\sum_{g=1}^G m_g(m_g-1)}.
\end{align*}
Also define \(\hat{\kappa} = \frac{n+N}{G}\).

Suppose that \(G',\mu'_p,\eta'_p\) are the design parameters for a new study that uses data from the same population as the pilot study. Then, the in-sample version of equations \eqref{eq:design-var} and \eqref{eq:design-corr} are
\begin{align}\label{eq:design-estimates}
    \hat{V}^\PPI &= \frac{1}{\hat{\kappa}^2 G'\mu'_p}\left(\hat{V}_\varepsilon-\frac{\hat{V}_{\varepsilon,Z}^2}{\hat{V}_Z}\right)\\
    \hat{\rho} &= \frac{\hat{V}_{\varepsilon,Z}}{\sqrt{\hat{V}_\varepsilon\hat{V}_Z}},
\end{align}
where
\begin{align*}
    \hat{V}_\varepsilon &= \hat{V}_{Y,s}+\left(\eta'_p(1-\mu'_p)+\mu'_p\right)\hat{V}_{Y,d},\\
    \hat{V}_Z &=\frac{1}{1-\mu'_p}\left(\hat{V}_{f,s}+
    \eta'_p\hat{V}_{f,d}\right), \text{ and}\\
    \hat{V}_{\varepsilon,Z}&=\hat{V}_{Y,f,s} + \eta'_p \hat{V}_{Y,f,d}.
\end{align*}
The accuracy of these forecasts depends on both the size of the pilot and how its labels are distributed. Estimating \(V_{Y,d}\), which captures how outcomes vary together within clusters, requires at least two labeled units in some clusters. Small pilots can yield unstable moment estimates, producing negative variance forecasts or estimated correlations outside \([-1,1]\). Such values signal that the corresponding forecasts are unreliable.

\section{Additional details on the empirical application}\label{app:application}

\subsection{Data}\label{app:data}

We combine two data sets from the replication data of \citet{girbau_face_2024}. The first contains the human annotations. In each of the $67{,}062$ sampled frames, the annotators drew a box around the face of each of the 41 politicians on screen and labeled the box with the politician's name. The second contains the output of the face recognition system, which the authors ran on all broadcasts since 2013, from start to finish, to measure the politicians' screen time. The system analyzes one frame per second, or about every 30th frame, and the output records which of the 41 politicians it identified in each analyzed frame. Frames in which it identified none of them do not appear in the output. The sitting prime minister, whose appearance defines $Y_{g,i}$, is Shinz\={o} Abe in broadcasts until 15 September 2020, Yoshihide Suga from 16 September 2020, and Fumio Kishida from 4 October 2021.

To pair each annotated frame with a prediction, we search the output for the frame of the same broadcast that is closest to it, provided it is at most 14 frames, or about half a second, away. The frames that the system analyzed are one second apart, so at most one of them lies this close to the annotated frame, namely the one from the same second. We set $f(X_{g,i})=1$ if the system identified the sitting prime minister in that frame, and $f(X_{g,i})=0$ if it identified only other politicians or if the output lists no frame this close. For example, in the NHK News 7 broadcast of 12 March 2018, the annotators marked Shinz\={o} Abe in frame 32926, so $Y_{g,i}=1$. The closest frame in the output, 32938, lies 12 frames later, and the system identified Abe in it, so $f(X_{g,i})=1$. The annotators also marked Abe one second earlier, in frame 32896. The output lists no frame within 14 frames of it, so the system identified no politician in that second and $f(X_{g,i})=0$. In most pairs, the frame in the output is the annotated frame itself, and only 3\% of the annotated frames are paired with a frame a few frames away, as in this example.

\subsection{Simulation details}\label{app:sim-details}

The mean labeling probability $\mu_p$ in Section~\ref{sec:sim-design} equals the expected share of labeled frames. The realized share varies across replications, more so in the fully labeled design, where it depends on how many broadcasts are labeled and how long they are. At $\mu_p=0.05$, for example, the realized share lies between 3.8\% and 6.3\% in 95\% of replications in the fully labeled design, but between 4.8\% and 5.2\% in the partially labeled design.

The two designs also produce different types of clusters (Figure~\ref{fig:cluster-data}). The fully labeled design yields only fully labeled and fully unlabeled broadcasts. In the partially labeled design, short broadcasts can end up with no labeled frame when $\mu_p$ is small and with every frame labeled when $\mu_p$ is large. About 38\% of broadcasts are fully unlabeled at $\mu_p=0.05$, and about 38\% are fully labeled at $\mu_p=0.95$.

\subsection{Full results}\label{app:telestats}

Table~\ref{tab:telestats-all} reports coverage and RMSE for every value of $\mu_p$ in both designs. The labels-only estimator averages observed outcomes over labeled frames, while the predictions-only estimator averages predictions over all sampled frames. Both use standard errors clustered by broadcast.

\begin{table}[h]
\centering
\small
\caption{Coverage of nominal 95\% intervals for $\theta^\star=0.109$ and RMSE in percentage points, by labeling design and mean labeling probability $\mu_p$. The predictions alone (Predictions only) use no labels, so their values differ across rows only by simulation noise. 2{,}500 replications per row.}
\label{tab:telestats-all}
\vspace{4pt}
\begin{tabular}{lcccccccc}
\toprule
 & \multicolumn{4}{c}{Coverage} & \multicolumn{4}{c}{RMSE} \\
\cmidrule(lr){2-5}\cmidrule(lr){6-9}
 & Labels & Predictions & PPI++ & Cluster-Robust & Labels & Predictions & PPI++ & Cluster-Robust \\
$\mu_p$ & only & only & (i.i.d.) & PPI++ & only & only & (i.i.d.) & PPI++ \\
\midrule
\multicolumn{9}{l}{\emph{Fully labeled design}} \\
0.05 & 0.90 & 0.18 & 0.27 & 0.88 & 3.18 & 1.98 & 1.61 & 1.36 \\
0.15 & 0.94 & 0.21 & 0.28 & 0.93 & 1.85 & 1.93 & 0.99 & 0.94 \\
0.25 & 0.94 & 0.19 & 0.26 & 0.94 & 1.41 & 1.95 & 0.85 & 0.83 \\
0.35 & 0.93 & 0.19 & 0.27 & 0.94 & 1.21 & 1.95 & 0.81 & 0.81 \\
0.45 & 0.95 & 0.20 & 0.28 & 0.95 & 1.03 & 1.95 & 0.75 & 0.75 \\
0.55 & 0.94 & 0.20 & 0.26 & 0.94 & 0.96 & 1.95 & 0.75 & 0.76 \\
0.65 & 0.95 & 0.20 & 0.25 & 0.94 & 0.88 & 1.96 & 0.73 & 0.73 \\
0.75 & 0.95 & 0.20 & 0.27 & 0.94 & 0.81 & 1.93 & 0.72 & 0.73 \\
0.85 & 0.95 & 0.20 & 0.27 & 0.94 & 0.76 & 1.95 & 0.72 & 0.73 \\
0.95 & 0.95 & 0.19 & 0.26 & 0.95 & 0.71 & 1.93 & 0.69 & 0.70 \\
\addlinespace
\multicolumn{9}{l}{\emph{Partially labeled design}} \\
0.05 & 0.95 & 0.20 & 0.54 & 0.95 & 0.88 & 1.94 & 0.73 & 0.73 \\
0.15 & 0.95 & 0.19 & 0.38 & 0.94 & 0.76 & 1.94 & 0.71 & 0.72 \\
0.25 & 0.95 & 0.20 & 0.33 & 0.94 & 0.73 & 1.94 & 0.70 & 0.71 \\
0.35 & 0.95 & 0.20 & 0.30 & 0.94 & 0.73 & 1.95 & 0.72 & 0.72 \\
0.45 & 0.95 & 0.20 & 0.29 & 0.95 & 0.71 & 1.93 & 0.70 & 0.70 \\
0.55 & 0.95 & 0.21 & 0.28 & 0.94 & 0.72 & 1.93 & 0.72 & 0.72 \\
0.65 & 0.95 & 0.19 & 0.28 & 0.95 & 0.70 & 1.95 & 0.70 & 0.70 \\
0.75 & 0.95 & 0.20 & 0.28 & 0.95 & 0.70 & 1.92 & 0.69 & 0.69 \\
0.85 & 0.94 & 0.20 & 0.26 & 0.94 & 0.72 & 1.93 & 0.72 & 0.73 \\
0.95 & 0.95 & 0.20 & 0.28 & 0.95 & 0.69 & 1.93 & 0.69 & 0.69 \\
\bottomrule
\end{tabular}
\end{table}

\clearpage
\section{Application to multiple regression}\label{app:regression}

\subsection{Simulation design}\label{app:reg-design}

Section~\ref{sec:experiments} estimates a mean, but many studies estimate the coefficients of a regression. Here we demonstrate Cluster-Robust PPI++ for a multiple regression, using the data of Section~\ref{sec:sim-design}. The broadcasts come from two news programs, NHK News 7 on the public broadcaster NHK and Hodo Station on the commercial network TV Asahi. We regress the indicator that the sitting prime minister is on screen on two characteristics of the broadcast. The first, $\mathrm{Era}_g$, equals one if broadcast $g$ aired in the Suga and Kishida era, from 16 September 2020, and zero if it aired under Abe. The second, $\mathrm{NHK}_g$, equals one if the broadcast comes from NHK News 7 and zero if it comes from Hodo Station. The regression is
\[
Y_{g,i}=\beta_0+\beta_1\,\mathrm{Era}_g+\beta_2\,\mathrm{NHK}_g+\varepsilon_{g,i}.
\]
The targets are the coefficients of this regression fitted to all human labels in the study population (Table~\ref{tab:regression-truth}). We use the regression to evaluate the method, not to test a substantive hypothesis. The era coefficient describes a difference between two periods. It does not isolate the effect of a new prime minister, because the Suga and Kishida era coincides with the COVID-19 pandemic.

With all frames labeled, the prime minister is on screen 13.4 percentage points more often in the Suga and Kishida era than under Abe. Fitted to the predictions instead, the regression puts this gap at 11.1 percentage points, 17\% lower, because the predictions understate the prime minister's screen time by the same share, about a sixth, in both eras. The NHK coefficient is nearly the same with the predictions (2.18 rather than 2.23 percentage points), because the predictions understate screen time by about the same number of percentage points in both programs. Controlling for the program leaves the era coefficient essentially unchanged.

\begin{table}[h]
\centering
\small
\caption{Least-squares regression of whether the sitting prime minister is on screen, fitted to all $67{,}062$ frames of the study population. Coefficients in percentage points, with standard errors clustered by broadcast in parentheses. The coefficients from the human labels are the targets of the simulation.}
\label{tab:regression-truth}
\vspace{4pt}
\begin{tabular}{lcc}
\toprule
 & Human labels, $Y_{g,i}$ & Predictions, $f(X_{g,i})$ \\
\midrule
Suga and Kishida era & 13.39 & 11.14 \\
 & (2.74) & (2.58) \\
NHK News 7 & 2.23 & 2.18 \\
 & (1.35) & (1.24) \\
Intercept & 7.85 & 6.37 \\
 & (0.89) & (0.79) \\
\bottomrule
\end{tabular}
\end{table}

We evaluate the estimators as in Section~\ref{sec:sim-design}, with the same designs, values of $\mu_p$, and random draws. The labels-only estimator is least squares on the labeled frames with cluster-robust standard errors. PPI++ and Cluster-Robust PPI++ use the general estimator of Section~\ref{sec:general} with the squared-error loss $\ell_\theta(x,y)=\frac{1}{2}(y-x^\top\theta)^2$, and each chooses a single $\hat\lambda$ to minimize the sum of its estimated variances of the three coefficients. We also fit the regression to the predictions of all sampled frames, ignoring the labels, with cluster-robust standard errors.

\subsection{Results on coverage and RMSE}\label{app:reg-results}

Figure~\ref{fig:regression-era} and Table~\ref{tab:regression-era} show coverage and RMSE for the era coefficient, and Figure~\ref{fig:regression-nhk} and Table~\ref{tab:regression-nhk} show them for the NHK coefficient. The results mirror those for the mean. PPI++ with i.i.d.\ standard errors covers the era coefficient in only 21\% to 51\% of replications and the NHK coefficient in only 26\% to 54\%. Cluster-Robust PPI++ covers the NHK coefficient in 94\% to 96\% of replications in every setting, and the era coefficient in 91\% to 95\% in every setting but one. In the fully labeled design with $\mu_p=0.05$, only about 15 Suga and Kishida broadcasts are labeled, and both Cluster-Robust PPI++ (81\%) and the labels-only estimator (81\%) undercover the era coefficient. For both coefficients, Cluster-Robust PPI++ reduces the RMSE of the labels-only estimator in the fully labeled design most when labels are scarce, by about half for $\mu_p$ up to $0.15$, and the reduction shrinks steadily to about 10\% at $\mu_p=0.75$.

How far the predictions alone can be trusted depends on the coefficient. Fitted to the predictions alone, the regression understates the era coefficient by about 2.3 percentage points, and its intervals cover the target in only about 82\% of replications. Because the era coefficient is estimated imprecisely, its RMSE is still lower than that of Cluster-Robust PPI++ when 15\% or fewer of the broadcasts are labeled in the fully labeled design. For the NHK coefficient, the predictions alone are nearly unbiased, so their intervals cover the target in about 95\% of replications and their RMSE is lower than that of every estimator that uses labels. Without labels, a researcher could not tell which of these cases applies.

\begin{figure}[h]
\centering
\includegraphics[width=\textwidth]{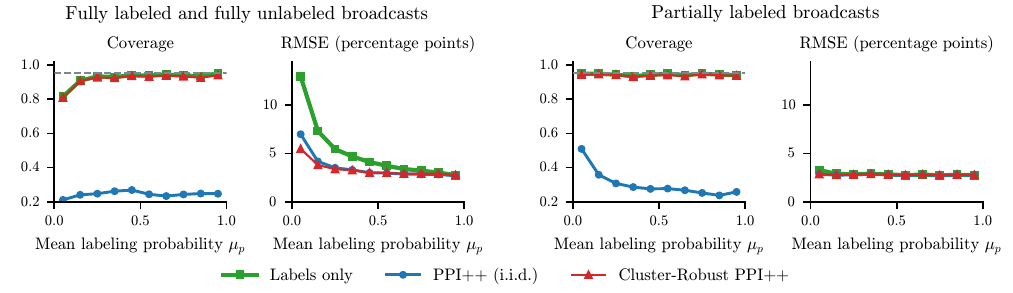}
\caption{Coverage, the share of 95\% confidence intervals that contain the era coefficient $\beta_1=13.39$ percentage points, and RMSE in percentage points, by the mean labeling probability $\mu_p$. In the fully labeled design (left), each broadcast is fully labeled or fully unlabeled. In the partially labeled design (right), most broadcasts are partially labeled. 2{,}500 replications per value of $\mu_p$.}
\label{fig:regression-era}
\end{figure}

\begin{table}[h]
\centering
\small
\caption{Coverage of nominal 95\% intervals for the era coefficient $\beta_1=13.39$ percentage points and RMSE in percentage points, by labeling design and mean labeling probability $\mu_p$. The predictions alone (Predictions only) use no labels, so their values differ across rows only by simulation noise. 2{,}500 replications per row.}
\label{tab:regression-era}
\vspace{4pt}
\begin{tabular}{lcccccccc}
\toprule
 & \multicolumn{4}{c}{Coverage} & \multicolumn{4}{c}{RMSE} \\
\cmidrule(lr){2-5}\cmidrule(lr){6-9}
 & Labels & Predictions & PPI++ & Cluster-Robust & Labels & Predictions & PPI++ & Cluster-Robust \\
$\mu_p$ & only & only & (i.i.d.) & PPI++ & only & only & (i.i.d.) & PPI++ \\
\midrule
\multicolumn{9}{l}{\emph{Fully labeled design}} \\
0.05 & 0.81 & 0.82 & 0.21 & 0.81 & 12.94 & 3.52 & 6.98 & 5.48 \\
0.15 & 0.91 & 0.83 & 0.24 & 0.91 & 7.32 & 3.36 & 4.17 & 3.80 \\
0.25 & 0.93 & 0.82 & 0.25 & 0.93 & 5.44 & 3.45 & 3.51 & 3.37 \\
0.35 & 0.93 & 0.82 & 0.26 & 0.92 & 4.69 & 3.50 & 3.29 & 3.25 \\
0.45 & 0.94 & 0.82 & 0.27 & 0.94 & 4.11 & 3.41 & 3.00 & 3.01 \\
0.55 & 0.94 & 0.82 & 0.24 & 0.93 & 3.72 & 3.46 & 2.96 & 3.00 \\
0.65 & 0.94 & 0.83 & 0.23 & 0.94 & 3.39 & 3.40 & 2.87 & 2.91 \\
0.75 & 0.94 & 0.83 & 0.24 & 0.93 & 3.22 & 3.42 & 2.85 & 2.90 \\
0.85 & 0.94 & 0.82 & 0.25 & 0.92 & 3.02 & 3.47 & 2.82 & 2.90 \\
0.95 & 0.95 & 0.82 & 0.25 & 0.94 & 2.75 & 3.43 & 2.70 & 2.77 \\
\addlinespace
\multicolumn{9}{l}{\emph{Partially labeled design}} \\
0.05 & 0.95 & 0.83 & 0.51 & 0.94 & 3.29 & 3.39 & 2.84 & 2.86 \\
0.15 & 0.95 & 0.83 & 0.36 & 0.95 & 2.89 & 3.40 & 2.76 & 2.78 \\
0.25 & 0.94 & 0.82 & 0.31 & 0.94 & 2.84 & 3.45 & 2.78 & 2.79 \\
0.35 & 0.93 & 0.81 & 0.29 & 0.93 & 2.91 & 3.49 & 2.86 & 2.87 \\
0.45 & 0.94 & 0.83 & 0.28 & 0.94 & 2.83 & 3.41 & 2.79 & 2.80 \\
0.55 & 0.95 & 0.83 & 0.28 & 0.94 & 2.75 & 3.39 & 2.73 & 2.75 \\
0.65 & 0.94 & 0.82 & 0.27 & 0.93 & 2.82 & 3.49 & 2.81 & 2.82 \\
0.75 & 0.95 & 0.83 & 0.25 & 0.95 & 2.74 & 3.38 & 2.73 & 2.73 \\
0.85 & 0.94 & 0.82 & 0.24 & 0.94 & 2.79 & 3.48 & 2.78 & 2.80 \\
0.95 & 0.94 & 0.82 & 0.26 & 0.94 & 2.75 & 3.43 & 2.75 & 2.76 \\
\bottomrule
\end{tabular}
\end{table}

\begin{figure}[h]
\centering
\includegraphics[width=\textwidth]{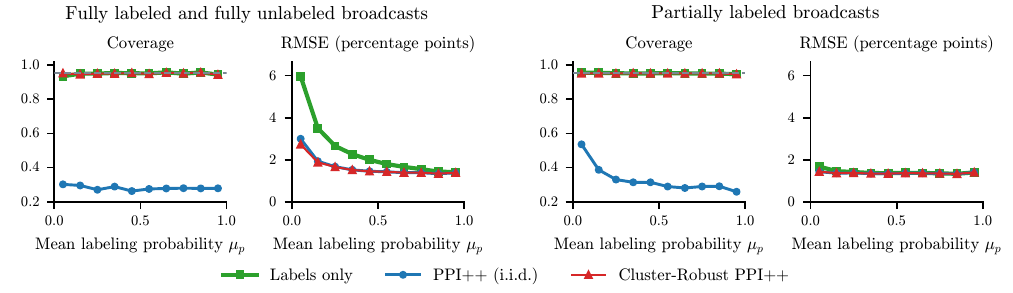}
\caption{Coverage, the share of 95\% confidence intervals that contain the NHK coefficient $\beta_2=2.23$ percentage points, and RMSE in percentage points, by the mean labeling probability $\mu_p$. In the fully labeled design (left), each broadcast is fully labeled or fully unlabeled. In the partially labeled design (right), most broadcasts are partially labeled. 2{,}500 replications per value of $\mu_p$.}
\label{fig:regression-nhk}
\end{figure}

\begin{table}[h]
\centering
\small
\caption{Coverage of nominal 95\% intervals for the NHK coefficient $\beta_2=2.23$ percentage points and RMSE in percentage points, by labeling design and mean labeling probability $\mu_p$. The predictions alone (Predictions only) use no labels, so their values differ across rows only by simulation noise. 2{,}500 replications per row.}
\label{tab:regression-nhk}
\vspace{4pt}
\begin{tabular}{lcccccccc}
\toprule
 & \multicolumn{4}{c}{Coverage} & \multicolumn{4}{c}{RMSE} \\
\cmidrule(lr){2-5}\cmidrule(lr){6-9}
 & Labels & Predictions & PPI++ & Cluster-Robust & Labels & Predictions & PPI++ & Cluster-Robust \\
$\mu_p$ & only & only & (i.i.d.) & PPI++ & only & only & (i.i.d.) & PPI++ \\
\midrule
\multicolumn{9}{l}{\emph{Fully labeled design}} \\
0.05 & 0.93 & 0.95 & 0.30 & 0.95 & 5.96 & 1.25 & 3.00 & 2.73 \\
0.15 & 0.95 & 0.95 & 0.30 & 0.94 & 3.50 & 1.27 & 1.93 & 1.88 \\
0.25 & 0.95 & 0.94 & 0.27 & 0.95 & 2.66 & 1.26 & 1.67 & 1.66 \\
0.35 & 0.95 & 0.95 & 0.29 & 0.95 & 2.27 & 1.24 & 1.52 & 1.51 \\
0.45 & 0.95 & 0.95 & 0.26 & 0.95 & 2.00 & 1.24 & 1.45 & 1.45 \\
0.55 & 0.95 & 0.95 & 0.28 & 0.95 & 1.80 & 1.25 & 1.43 & 1.43 \\
0.65 & 0.96 & 0.95 & 0.28 & 0.95 & 1.66 & 1.22 & 1.39 & 1.39 \\
0.75 & 0.95 & 0.95 & 0.28 & 0.95 & 1.55 & 1.25 & 1.39 & 1.39 \\
0.85 & 0.96 & 0.95 & 0.28 & 0.96 & 1.44 & 1.21 & 1.34 & 1.34 \\
0.95 & 0.94 & 0.95 & 0.28 & 0.94 & 1.41 & 1.25 & 1.38 & 1.38 \\
\addlinespace
\multicolumn{9}{l}{\emph{Partially labeled design}} \\
0.05 & 0.95 & 0.95 & 0.54 & 0.95 & 1.67 & 1.23 & 1.43 & 1.43 \\
0.15 & 0.95 & 0.95 & 0.39 & 0.95 & 1.46 & 1.24 & 1.36 & 1.37 \\
0.25 & 0.95 & 0.95 & 0.33 & 0.95 & 1.42 & 1.24 & 1.37 & 1.37 \\
0.35 & 0.95 & 0.95 & 0.31 & 0.95 & 1.37 & 1.22 & 1.34 & 1.34 \\
0.45 & 0.95 & 0.95 & 0.31 & 0.95 & 1.36 & 1.22 & 1.34 & 1.34 \\
0.55 & 0.95 & 0.94 & 0.29 & 0.95 & 1.38 & 1.25 & 1.36 & 1.37 \\
0.65 & 0.95 & 0.95 & 0.28 & 0.95 & 1.38 & 1.25 & 1.36 & 1.36 \\
0.75 & 0.95 & 0.95 & 0.29 & 0.95 & 1.35 & 1.22 & 1.35 & 1.35 \\
0.85 & 0.95 & 0.95 & 0.29 & 0.95 & 1.33 & 1.22 & 1.33 & 1.33 \\
0.95 & 0.95 & 0.94 & 0.26 & 0.95 & 1.40 & 1.28 & 1.40 & 1.40 \\
\bottomrule
\end{tabular}
\end{table}

\end{document}